\documentclass[bj]{imsart}
\usepackage{amsmath, amsthm, amssymb}
\usepackage[english]{babel}
\usepackage{dsfont}
\usepackage{mathdots}
\usepackage{natbib}
\usepackage{graphicx}
\usepackage{enumerate}
\usepackage{color}
\usepackage{float}
\usepackage{placeins}
\usepackage{todonotes}
\usepackage{multirow}
\usepackage{comment}
\usepackage{url} 
\arxiv{arXiv:0000.0000}

\newtheorem{thm}{{Theorem}}
\newtheorem{lem}[thm]{{Lemma}}

\newtheorem{asp}{{Assumption}}

\newtheorem{corollary}[thm]{{Corollary}}
\newtheorem{example}[thm]{{Example}}

\newcommand{\argmin}{\mbox{argmin}}
\newcommand{\var}{\mbox{Var}}
\newcommand{\cov}{\mbox{Cov}}

\newcommand{\kp}{k_{1}(p)}
\newcommand{\kpp}{k_{3}(p)}
\newcommand{\cp}{k_{2}(p)}
\newcommand{\kppp}{\cp \kp^4\kpp^2}
\newcommand{\gp}{g(p,d,n)}

\newcommand{\Sigmaeps}{\Sigma_\eps}
\def\eps{\varepsilon}

\newcommand{\R}{\mathds{R}}

\newcommand{\N}{\mathds{N}}
\newcommand{\Z}{\mathds{Z}}
\newcommand{\ind}{\mathds{1}}
\newcommand{\pt}{{\tilde {p}}}
\newcommand{\Gammah}{\hat{ {\Gamma}}^{(st)}(0)}
\newcommand{\Gammahast}{\hat{ {\Gamma}}^{*(st)}(0)}

\newcommand{\U}{\mathds{U}}
\newcommand{\A}{\mathds{A}}
\newcommand{\G}{{n_g}}
\newcommand{\Gn}{{G_n}}
\newcommand{\Gammas}{{\Gamma}^{(st)}(0)}
\newcommand{\Gammasast}{{\Gamma}^{*(st)}(0)}
\newcommand{\Sigmah}{\hat \Sigma_\eps}
\newcommand{\Sigmahthr}{\hat \Sigma_{\hat \eps}}

\begin{document}

\begin{frontmatter}

\title{Bootstrap Based Inference for Sparse High-Dimensional Time Series Models}
\runtitle{Bootstrap Based Inference for Sparse High-Dimensional Time Series Models}

\begin{aug}
  \author{\fnms{Jonas}  \snm{Krampe}\thanksref{a}{j.krampe@uni-mannheim.de}},
  \author{\fnms{Jens-Peter} \snm{Kreiss}\thanksref{b}{j.kreiss@tu-bs.de}}
  \and
  \author{\fnms{Efstathios }  \snm{Paparoditis}\thanksref{c}{stathisp@ucy.ac.cy}}

  \runauthor{Krampe et al.}

  \affiliation{University of Mannheim,Technische Universit\"at Braunschweig and University of Cyprus}

  \address[a]{University of Mannheim}

  \address[b]{Technische Universit\"at Braunschweig}
  
  \address[c]{University of Cyprus}

\end{aug}

\begin{abstract}\
Fitting sparse models to high-dimensional time series is an important area of statistical inference.
In this paper we consider sparse vector autoregressive models and develop appropriate bootstrap methods to infer properties of such processes.
Our bootstrap methodology generates pseudo time series using a model-based bootstrap procedure which involves an estimated, sparsified version of the underlying vector autoregressive model.
Inference is performed using so-called de-sparsified or de-biased estimators of the autoregressive model parameters. We derive the asymptotic distribution of such estimators in the time series context and establish asymptotic validity of the bootstrap procedure proposed for estimation and, appropriately modified, for testing purposes. In particular we focus on testing that  large
groups of autoregressive coefficients equal zero. Our theoretical results are complemented by simulations which investigate the finite sample performance of the bootstrap methodology proposed. A real-life data application is also presented.
\end{abstract}

\end{frontmatter}

\def\spacingset#1{\renewcommand{\baselinestretch}%
{#1}\small\normalsize} \spacingset{1}

\section{Introduction}

Statistical analysis of high-dimensional time series has attracted considerable interest during the last decades. Initiated by developments in the i.i.d., mainly   regression, set-up,  statistical methods have been proposed to select  and to estimate non-zero  parameters in the context of sparse  high-dimensional time series models by means of regularized-type  estimators. To be more specific,  consider  
a $p$ dimensional stochastic process $ \{X_t, t\in \Z\}$,  where the random vector  $X_t$ is generated via a $d$th order     vector autoregressive (VAR($d$)) model,  
\begin{equation} \label{eq.var}
X_t=\sum_{s=1}^d A^{(s)} X_{t-s}+ \eps_t. 
\end{equation}
Here $A^{(s)}$, $s=1,2, \ldots, d$,  are $p\times p$ coefficient matrices while the  $\eps_t$'s are  assumed to be independent and identically distributed (i.i.d.) innovations with $ E(\eps_t)=0, \var(\eps_t)=\Sigma_\eps$, in short,  $\eps_t\sim (0,\Sigma_\eps)$. Assume that the process is stationary and causal, that is  $\det(\mathcal{A}(z))\not = 0$ for all $z\leq 1$, where 
 $\mathcal{A}(z)=I-\sum_{s=1}^d A^{(s)} z^s$. Model (\ref{eq.var})
  has $dp^2$ unknown  parameters  in the matrices $A^{(s)}, s =1,\dots,d$ and  $p(p+1)/2$  unknown parameters in the innovation covariance matrix  $\Sigma_\eps$. Hence the total number of unknown parameters is  $q=p^2(d+1/2)+p/2=O(dp^2)$.  Suppose that a time series $ X_1, X_2, \ldots, X_n$ stemming from $\{X_t, t\in\Z\}$ has been observed. If the number of parameters is small in the sense that $q \ll n$, 
  then inference for such processes is a well developed and well understood area in multivariate  time series analysis; see among others,  \cite{reinsel2003elements,luetkepohl2007new,tsay2013multivariate} and \cite{kilian2017structural}.
  
  In this paper we consider the important case  where  $q \gg n$ but   the VAR(d)  model  (\ref{eq.var}) 
  possesses some form of sparse representation, that is  many of the parameter coefficients  are equal to zero.  To elaborate, we first  fix some notation. For  a random variable $ X$ we write   $ \|X\|_{E,q}$ for $\big(E|X|^q\big)^{1/q}$, where   $q\in \N$;  for  a vector $x\in \R^p$,  $\|x\|_0 = \sum_{j=1}^p \ind(x_j \not = 0)$, $ \| x \|_1 = \sum_{j=1}^p |x_j|$ and $\| x \|_2^2 = \sum_{j=1}^p |x_j|^2$.  Furthermore, for a $r\times s$ matrix $B=(b_{i,j})_{i=1,\ldots,r, j=1,\ldots,s}$,  $\|B\|_1=\max_{1\leq j\leq s}\sum_{i=1}^r|b_{i,j}|=\max_j \| B e_j\|_1$, 
   $\|B\|_\infty=\max_{1\leq i\leq r}\sum_{j=1}^s|b_{i,j}|=\max_{i} \| e_i^\top B\|_1$, where $e_j=(0,\ldots,0,1,0,\ldots, 0)^\top$ denotes  the  vector with the one appearing in the $j$th position. Denote the largest eigenvalue of a matrix $A$ by $\rho(A)$ and $\|A\|_2^2=\rho(AA^\top)$.
  Using this notation,  let $k_j^r(p)=\sum_{s=1}^d\|e_j^\top A^{(s)}\|_0$ and $k_j^c(p)=\sum_{s=1}^d\| A^{(s)}e_j\|_0$ be the number of non-zero coefficients in the $j$th row, respectively,  in the $j$th column of the matrices $ A^{(s)}$, $s=1,2, \ldots, d$, and, let  $\kp=\max_{1\leq j \leq p} \{k_j^r(p),k_j^c(p)\}$. 
  In the following we 
  consider the case  where  the VAR(d) model is sparse, that is the total number of non-zero coefficients within a row or column satisfies $\kp \ll n$. Furthermore, we allow $ \kp$ to depend on $p$, that is, $\kp$  can be an increasing function of the dimension $p$ of the process under consideration. 
  
 In the setting described above,  procedures to fit  sparse VAR models have been considered by many authors in the literature  by means of regularized type estimators; see among others, 
\cite{song2011large} for $\ell_1$-penalized least squares (lasso) estimators, \cite{han2015direct} for $\ell_1$-penalized Yule-Walker estimators,   \cite{kock2015oracle} for oracle type inequalities for  adaptive lasso estimators and \cite{davis2016sparse} for a two step procedure which includes $\ell_1$-penalized  likelihood estimators. 
 Consistency of  $\ell_1$-penalized  estimators  has been established for a specific sparsity setting by  \cite{basu2015},  while   \cite{lin2017regularized} considered  estimation  for multi-block high-dimensional  VAR models including  testing of Granger-causality.  However, and despite the progress made in fitting sparse VAR models,  statistical inference  for such models seems to be a   less developed area. 
This is probably due to the fact that the asymptotic distribution  of $\ell_1$-penalized  estimators is  difficult to derive and   
statistical inference is  much more involved and difficult to implement.  Notice that even in the  i.i.d. regression set-up  with fixed dimension, the limiting distribution of regularized lasso estimators has been shown to 
 be nonstandard and  one which  assigns  positive probability mass at zero to the zero coefficients; see \cite{knight2000asymptotics}. This 
 leads,  among other things,  
 to  inconsistency of standard,  model-based bootstrap methods;  \cite{chatterjee2010asymptotic}; see also   
\cite{chatterjee2011bootstrapping} for a different  consistent  bootstrap proposal in the high-dimensional  i.i.d. regression setting. 

In this paper we focus on the development of  bootstrap procedures for inferring  properties of  high-dimensional, sparse  VAR($d$) processes. Toward this  goal and in order to avoid potential   problems associated with nonstandard limiting behaviour of regularized lasso  estimators, we  propose to  bootstrap  de-biased or de-sparsified  estimators of the VAR parameters. De-sparsified estimators of sparse estimators obtained  by lasso regularization, have been introduced and investigated in the i.i.d. regression case  by  several  authors. We refer here to the initial paper by  \cite{zhang2014} and to   \cite{deGeer2014},  which investigated such estimators in a much more broader   setting  and established, under certain  regularity conditions,  asymptotic optimality in the sense of semi-parametric efficiency.  
In the same i.i.d. regression set-up, de-sparsified estimators  have also been used in the context of  model based bootstrap inference; 
see  \cite{dezeure2017} for a  discussion and a recent contribution. De-biased estimators in the context of Gaussian respectively sub-Gaussian VAR processes have also been considered in \cite{chaudhry2017uncertainty} and have been used for statistical inference and in particular for testing Granger causality. Furthermore, de-sparsified estimators (or versions of it) were used to do inference for systems of high-dimensional regression equations which included VAR systems.  \cite{chernozhukov2018lasso} show a near oracle inequality for the lasso estimator under a rather general set-up which goes beyond sub-Gaussian innovations. They also do inference by 
using the bootstrap but  they focus on simultaneous equation systems; see in particular  Example 2 in \cite{chernozhukov2018lasso}. Thus the inference set-up consider by these authors  does not include VAR processes. \cite{neykov2018unified} discusses a de-sparsifying approach for Gaussian systems. They construct an influence function by projecting the fitted estimating equations to  sparse directions. The resulting  de-sparsifying approach is more general in the sense that it can be also applied to nonlinear models.
For linear models, like VAR models,  estimation of these sparse directions correspond to the CLIME estimator, see \cite{cai2011constrained}. Hence, besides  Gaussianity, an important assumption of this approach is that the inverse lag-zero autocovariance matrix of the (stacked) VAR system is sparse. This sparsity is also assumed in \cite{zheng2018testing} who adapt the de-sparsifying approach of \cite{ning2017general} to VAR processes. The  innovations are assumed to be componentwise independent, that is,  their covariance matrix   is diagonal. As  we will see  in Section~\ref{subsec2.2},  a sparse VAR system does not necessarily imply that the inverse lag-zero autocovariance matrix is sparse. Hence, such an  assumption might be more restrictive in the  time series context  than it is in the i.i.d. regression set up. We propose here a de-sparsifying approach which makes use of the underlying VAR structure and does not required  sparsity of the inverse of the  lag-zero autocovariance matrix of the process.

 Our work extends and generalizes the aforementioned contributions in many directions. In particular, we  consider de-sparsified, respectively de-biased, estimators for the general VAR($d$) process  and we derive their limiting distribution under quite general  conditions on the process and on the stochastic properties of the, not necessarily Gaussian, innovations. We impose sparsity assumptions only  on the VAR coefficient matrices $A^{(s)}$ and on the innovation covariance matrix  (or on its inverse), but we {\it do not} impose sparsity assumptions directly  on the autocovariance structure of the VAR($d$) process.  Furthermore, we introduce a novel and valid bootstrap procedure for inferring properties of the parameters of the VAR($d$) process. Appropriately modified, this bootstrap procedure also allows,  for  testing statistical hypotheses about groups of model parameters  in  a very flexible way, where the  total number of hypotheses tested is also allowed to increase to infinity  with the sample size $n$.

We first derive  the limiting  distribution of    de-sparsified estimators for the parameters of a general, stationary  VAR($d$) process. We show that this limiting distribution is a regular Gaussian distribution which is   solely affected by the autocovariance structure of the underlying VAR($d$) process.  We then propose a bootstrap procedure to estimate the distribution of the  de-sparsified estimators of the VAR parameters. This procedure uses  an appropriately thresholded $\ell_1$-penalized  estimator of the coefficient matrices $A^{(s)}$ and a thresholded, sparse  estimator of the covariance matrix $\Sigmaeps$ of the innovations. Thresholding  is important in this context, since it guaranties  sparsity of the  VAR model used in the bootstrap world.  The  fitted sparse VAR model  driven  by appropriately generated 
i.i.d. innovations is then used to generate    vector pseudo time series $X_1^\ast, X_2^\ast, \ldots, X_n^\ast$, which appropriately  imitate the sparse stochastic structure of the observed time series $X_1, X_2, \ldots, X_n$.     We prove consistency of such a bootstrap procedure under general conditions when   applied  to estimate  the distribution of  de-sparsified estimators.  The results obtained allow for using  the bootstrap procedure proposed in order to construct  individual or simultaneous   confidence intervals and to perform tests of  hypotheses about  model parameters. In particular,  
 we show how the bootstrap method proposed can be used in order to test the interesting  hypothesis  that  individual or, more importantly,  groups of  coefficients  of the VAR model are equal to zero. 
  For such testing purposes,  the   bootstrap procedure is appropriately modified so that  the sparse VAR model used to generated the pseudo time series $X_1^\ast, X_2^\ast, \ldots, X_n^\ast$ satisfies the null hypothesis of interest. Consistency of the bootstrap based testing procedure  is established for  max-type  test statistics. Finally,  we demonstrate by means of numerical investigations that the theoretical results established are accompanied by a good finite sample behavior of the bootstrap methodology developed.

The  paper is organized as follows.  Section 2 
 introduces de-sparsified estimators for the VAR model parameters  and derives their limiting  distribution under suitable assumptions on the sparsity of the underlying VAR process and on the consistency properties of the regularized estimators involved. Section 3 introduces the bootstrap procedure proposed and establishes its asymptotic validity for estimating the distribution of de-sparsified  estimators
 and,  appropriately modified, for  testing  hypotheses  about  model parameters.  Asymptotic validity of the bootstrap based test procedure is established. Section 4 is devoted to  issues related to the practical implementation and to  numerical investigations. We propose a bias correction  procedure to improve the finite sample performance of the bootstrap based testing  and  we  present several simulations    supporting the  good size and power behavior of the bootstrap methodology proposed  for difficult inference problems.  An application to an interesting real-life data set is also discussed. Auxiliary results  and technical  proofs are deferred to Section \ref{sec.proofs}.

\section{De-Sparsified Estimators of VAR Parameters}
\subsection{De-sparsified estimators}  \label{subsec2.2}
In this section we present our approach to construct 
de-biased estimators in the high-dimensional time series context. Towards this  we  adapt to the time series set-up  the basic idea  for  the introduction of de-sparsified estimators used in the  i.i.d. regression context, (see \cite{dezeure2017}) and make the appropriate modifications. We first  fix  some additional notation. Let $I_{p,-j}$ be the $(p-1)\times p $ matrix obtained from  the  identity matrix of dimension $p$ after deleting its  $j$-th row. For a matrix $A \in \R^{q \times p}$,  $A_{j;r}:=e_j^\top A e_r$ is its $(j,r)$ element and  $A_{j;-r}:= e_j^\top A I_{p;-r}^\top$ its  $j$-th  row  after deleting the element $A_{j;r}$.

The VAR$(d)$ process given in \eqref{eq.var} can be written as $$X_t=(A^{(1)},\dots,A^{(d)}) (X_{t-1}^\top,\dots,X_{t-d}^\top)^\top+\eps_t=: \Xi W_{t-1} +\eps_t,$$
with an obvious notation for the matrices $\Xi\in \R^{p\times dp}$ and $ W_{t-1}\in \R^{dp\times 1}$. 

Observe first    that if $ p \ll n$, then  the standard least squares estimator of $\Xi_{j;r}$ obtained by regressing $X_{t;j}$ onto $W_{t-1;r}, r=1,\dots,dp$,  also can be written as 
\begin{align} \label{eq.standard-lse}
\widehat{\Xi}_{j;r} & =   \sum_{t=d+1}^n V_{t-1;r} X_{t;j}\big/ \sum_{t=d+1}^n V_{t-1;r}W_{t-1;r},
\end{align}
where $V_{t-1;r}$ are the ``residuals''  obtained as the difference between  $ X_{t-1;r}$ and its best (in mean square) linear approximation  using 
all other variables 
 contained in  the lagged vector $ X_{t-1}$. 
$ V_{t-1;r} $  also can   be expressed  as $ V_{t-1;r}= \tilde{\beta}^\top_r  W_{t-1}$,  where 
$ \tilde \beta_r= \big(e_r^\top \widetilde\Gamma^{(st)}(0)^{-1} e_r\big)^{-1}\widetilde{\Gamma}^{(st)}(0)^{-1}e_r$, with $ \widetilde{\Gamma}^{(st)}(0)$  
the sample covariance matrix of $\{W_t,t=d+1,\dots,n\}$ 
and which is 
given by $\tilde \Gamma(0)^{(st)}=(n-d)^{-1} \sum_{t=d+1}^n (W_t-\bar W_n) (W_t-\bar W_n)^\top$,  $\bar W_n=(n-d)^{-1} \sum_{t=d+1}^n W_t$; see also  Lemma~\ref{lem.dagger}.

Clearly,  in the high-dimensional set-up, that is if $ p \gg n$, such a construction  
is not  possible since in this case $   V_{t-1;r}\equiv 0$, due to the fact that the row  vector  $\{W_{t-1;r}, t=d+1,d+2, \ldots, n\}$ is an element of the subspace spanned 
by $\{W_{t-1;-r}, t=d+1,d+2,\dots,n\}.$  
To overcome this problem in the i.i.d. regression set-up, the approach followed     is to replace the estimator $ \widetilde{\beta}_r$ appearing  in the definition of the residuals $ V_{t-1;r}$, by some regularized estimator. This approach requires the imposition of  sparsity assumptions on the corresponding vector of coefficients $\beta_r=(e_r (\Gammas^{-1} e_r)^{-1} (\Gammas)^{-1} e_r$, that is on the inverse of the  covariance matrix $ \Gamma^{(st)}(0)^{-1}=E(W_t W_t^\top)$. However, in the  sparse VAR time series
set-up considered in this paper, such an approach seems not to be appropriate. The reason for this lies in the fact that due to  temporal  dependence,  sparsity of the matrices $ A^{(1)}, \ldots, A^{(d)}$  and $\Sigma_\varepsilon$ respectively  $\Sigma^{-1}_\varepsilon$,  does
necessarily imply  sparsity of the  autocovariance matrices $ \Gamma^{(st)}(0)$
respectively $ \Gamma^{(st)}(0)^{-1}$. In fact, even in the simple  VAR(1) case, the corresponding model parameter matrices $A=A^{(1)}$ and $ \Sigma_\varepsilon$ may be sparse but $ \Gamma^{(st)}(0)$ respectively $ \Gamma^{(st)}(0)^{-1}$ may not. For the same case  observe that if $ A$ is symmetric then by properties of the Neumann-series it can be shown  that for the inverse of the covariance matrix $\Gamma(0)=E(X_tX^\top_t)$ it yields that   $ \Gamma(0)^{-1}=\Sigma_\varepsilon^{-1} -\Sigma_\varepsilon^{-1}A\Sigma_\varepsilon A\Sigma_\varepsilon^{-1}$, which implies certain restrictions on the sparsity conditions one has to impose on  $A$ and on $ \Sigma_\varepsilon$, respectively  $ \Sigma_\varepsilon^{-1}$, in order to achieve the  desired  sparsity properties of $ \Gamma(0)^{-1}$.   Thus imposing, in the sparse VAR time series set-up, sparsity assumptions directly on the autocovariance matrices respectively on  its inverse, is difficult to justify  and may implicitly restrict  the class of VAR(d) processes considered. For these reasons,   the  sparsity behavior of the VAR(d) process is handled  in this paper by    imposing   sparsity assumptions directly  on the process parameter matrices $ A^{(1)}, \ldots, A^{(d)}$  and on $\Sigma_\varepsilon$,  respectively,  on  $\Sigma_\varepsilon^{-1}$. Because of this,  and in order  to construct de-sparsified estimators, an alternative approach to the one used in the i.i.d. regression set-up has to be   followed.

To elaborate, recall that 
$\tilde\beta_r$ is an estimator of $\beta_r=(e_r (\Gammas^{-1} e_r)^{-1} (\Gammas)^{-1} e_r$
and that if $p \ll n$, then   least squares leads  to the estimation of $\Gammas$   by the sample covariance  $\tilde \Gamma^{(st)}(0)$.
Clearly $\tilde \Gamma^{(st)}(0)$ is not a consistent estimation of $\Gamma^{(st)}(0)$ (or of its inverse) if  $p > n$.  Thus  a  suitable estimator of  $\Gamma^{(st)}(0)$ has to be  used  which  satisfies certain consistency properties in the high-dimensional sparse time series setting considered. We will discuss later on the construction of such an estimator  which we will denote by $\Gammah$. 
Given such an estimator,   we can then construct our  ``residuals'' as $\hat Z_{t-1;r} = \hat \beta_r W_{t-1},$ where now $ \hat{\beta}_r$ is defined by  $\hat{\beta}_r =(e_r^\top \Gammah^{-1} e_r)^{-1} \Gammah^{-1} e_r.$  We  then have
\begin{align*}
\tilde \Xi_{j;r} =& \sum_{t=d+1}^n \! \hat Z_{t-1;r} X_{t;j}\big/ \sum_{t=d+1}^n \hat  Z_{t-1;r}W_{t-1;r}
=\Xi_{j;r}+\sum_{t=d+1}^n \hat  Z_{t-1;r} \eps_{t;j}\big/ \sum_{t=d+1}^n \hat  Z_{t-1;r}W_{t-1;r}  \\& +
\sum_{t=d+1}^n \hat  Z_{t-1;r} W_{t-1;-r} \Xi^\top e_j \big/ \sum_{t=d+1}^n \hat  Z_{t-1;r}W_{t-1;r}.    
\end{align*}
Since in the case $p>n$ considered, the ``residuals'' $\{Z_{t-1;r}\}$ are not perfectly orthogonal to $\{W_{t-1;-r}\}$, the last term above is non-zero. Consequently a bias  in the expression for  $\widetilde{\Xi}_{j;r} $ appears which, however,  can be  estimated using some (regularized) estimator of $\Xi=(A^{(1)},\dots,A^{(d)})$ which     satisfies certain consistency properties, to be discussed later on. Given such an estimator, which we denote by  $(\hat A^{(1,re)},\dots,\hat A^{(d,re)})=\hat\Xi^{(re)}$, we can  then estimate the bias term in  the last  displayed expression. Subtracting  the estimated  bias from $\tilde \Xi_{j;r}$  we  obtain the following de-biased estimator of $ A_{j;r}$ 
\begin{align}
\hat A_{j;r}^{(s,de)}  =\hat \Xi_{j;(s-1)d+r}^{(de)} =&\hat A_{j;r}^{(s,re)}+(\sum_{t=d+1}^n \hat Z_{t-1;(s-1)d+r} W_{t-1;((s-1)d+r)})^{-1} \nonumber \\ &  \  \times( \sum_{t=d+1}^n \hat Z_{t-1;(s-1)d+r} (X_{t;j}-e_j^\top \hat \Xi^{(re)} W_{t-1})), \label{de_est}
\end{align}
where  $r,j=1,\dots,p,s=1,\dots,d$. The estimator $\hat A_{j;r}^{(s,de)}$ given above is called a de-biased, respectively,   a de-sparsified estimator of $A_{j;r}^{(s)}$ since it is a bias corrected version  of the initial estimator $\tilde \Xi_{j;(s-1)d+r}=\tilde{A}_{j;r}^{(s)}$ and it is not sparse anymore. 

We now discuss   the estimators $\hat\Xi^{(re)} $ and  $\Gammah$  used in the above construction of the de-biased estimators $\hat A_{j;r}^{(s,de)}$. 
Suitable candidates for 
$\hat{\Xi}^{(re)}=(\hat A^{(1,re)},\dots,A^{(d,re)})$ in the sparse VAR$(d)$ setting are the  adaptive lasso estimator \cite{kock2015oracle}; see  
also \cite{chernozhukov2018lasso} and the discussion following equation (\ref{eq.init-thr}) below.   Such an estimator for the $i$th row of $\Xi$ is obtained as 
\begin{equation} \label{eq.adlasso} 
\hat{\Xi}_{i,\cdot}^{(re)} = \argmin_{c=(c_1, \ldots, c_dp)^\top\in\R^{dp}}\frac{1}{n-d}\sum_{t=d+1}^n\big(X_{i;t} -c^\top W_{t-1}\big)^2 + \lambda_n \sum_{s=1}^{dp} \frac{|c_s|}{|\hat{c}_{j,s}|},
\end{equation}
where $ \hat{c}_j=(\hat{c}_{j,s}, s=1,2, \ldots, dp)$ are the lasso estimators  of $ \Xi_{i,\cdot}$ obtained  as
$\hat{c}_j=\argmin_{c\in\R^{dp}}(n-d)^{-1}\sum_{t=d+1}^n\big(X_{i;t} -c^\top W_{t-1}\big)^2 + \lambda_n \|c\|_1$ and $\lambda_n = C\sqrt{\log(p)/n}$ is a regularized parameter. 

Estimation of  the covariance matrix $\Gammas$ is more involved. Recall   first that for a stable VAR(1) process with coefficient matrix $A=A^{(1)}$ we have $ \Gamma(0)=\sum_{j=0}^\infty A^j \Sigma_\eps (A^\top)^j=\operatorname{vec}^{-1}(I_{p^2}-A\otimes A)^{-1}\operatorname{vec}(\Sigma_\varepsilon)$.  This expression 
can be generalized to a VAR(d) process using its stacked VAR(1) representation, i.e.,  the representation $ W_t = \mathds{A} W_{t-1} +\mathds{U}_t$; see Appendix A for the definition of $\mathds{A} $ and   $\mathds{U} $. This representation leads to the expression  
\begin{equation}
    \label{eq.GammaHat}
 \Gamma^{(st)}(0)=\sum_{j=0}^\infty \mathds{A}^{j}\Sigma_{\mathds{U}} (\mathds{A}^\top)^j=\operatorname{vec}^{-1} \Big((I_{(dp)^2}-\mathds{A}\otimes\mathds{A})^{-1} \operatorname{vec}(\Sigma_{\varepsilon})\Big).
\end{equation}
The above  expression suggest that  a consistent estimator $\hat{\Gamma}^{(st)}(0)$  of $\Gamma^{(st)}(0)$ can be obtained  
by plugging in (\ref{eq.GammaHat})  estimators $\hat{A}^{(s)}$, $s=1,2, \ldots, d$  and $ \hat{\Sigma}_\varepsilon$  of the coefficient matrices $A^{(s)}$, $s=1,2, \dots,d$  and of the covariance matrix  $\Sigmaeps$, respectively, provided these estimators satisfy certain consistency properties.  

 To elaborate on the  construction of the estimator $ \hat{A}^{(s)}$,   as expression (\ref{eq.GammaHat}) shows, $\|\cdot\|_\infty $ consistency of  $\hat{\Gamma}^{(st)}(0)$  at a desired rate, requires consistency at an  appropriate rate of $\sum_{s=1}^p \|(\hat{A}^{(s)})^\top-(\A^{(s)})^\top\|_\infty= \sum_{s=1}^p \|\hat{A}^{(s)}-A^{(s)}\|_1 $ and of $\sum_{s=1}^p \|\hat{A}^{(s)}-A^{(s)}\|_\infty $. That is, we need to  control   the error made in estimating the coefficients $A_{j;r}^{(s)}$ in the  rows {\it and} in the columns of the matrices  $A^{(s)}$, $s=1, \ldots, d$,  taking into account that the dimension of these matrices is allowed to  increase to infinity with $n$. These considerations  motivate  the following estimator  of  $ A_{j;r}^{(s)}$:
  \begin{equation} \label{eq.init-thr}
    \hat A_{j;r}^{(s)}  = \hat A_{j;r}^{(s,adlasso)} \ind{\big(| \hat A_{j;r}^{(s,adlasso)} | \geq a_n,|\hat A_{j;r}^{(s,YW)} | \geq a_n \big)},
    \end{equation}
$j,r=1,\dots,p$, $s=1,\dots,d$. Here $\hat A_{j;r}^{(s,adlasso)}$ is the adaptive lasso estimator (\ref{eq.adlasso}), see also  \cite{kock2015oracle},  $\hat A_{j;r}^{(s,YW)}$ is the regularized  Yule-Walker type estimator, see  \cite{han2015direct} and  $ a_n =C\sqrt{\log(p)/n}$ is a   threshold parameter.  Notice that the regularized Yule-Walker estimator of the $j$th column of the matrix $\Xi$ is obtained as
\begin{equation} \label{eq.regRW}
\hat{\Xi}^{(YW)}_{\cdot,j}=\argmin_{c\in \R^{dp}, \|c\|_1} \|\hat{\Gamma}_{W}(0) c - \hat\Gamma_{W}(1)e_j \|_\infty \leq \lambda_n,
\end{equation}
where $\hat{\Gamma}_{W}(0)=(n-d)^{-1}\sum_{t=d}^n W_t W_t^\top$, 
$ \hat{\Gamma}_{W}(1)=(n-d)^{-1}\sum_{t=d}^{n-1} W_t W_{t+1}^\top$ and $\lambda_n>0$ is a regularization parameter.  To further  elaborate on the motivation leading to the estimator (\ref{eq.init-thr}), 
observe  that  it is obtained by using the combined support of two initial estimators,  that is of $\hat A_{j;r}^{(s,adlasso)}$ and of $\hat A_{j;r}^{(s,YW)} $. This is done in order to  achieve a desired   row- and columnwise $\ell_1$-consistency of the estimator $\hat{A}_{j;r}^{(s)}$ respectively of $ \hat A^{(s)}$, $s=1,\ldots, d$.  
In particular,  \cite{kock2015oracle} obtained  under some conditions which include Gaussian innovations, that   the row-wise $\ell_1$-error of the adaptive lasso has the rate $O_P(\kp \sqrt{(\log(n)\log(p)\log(d))^5/n})$. Recall that the adaptive lasso  is build up row-wise and therefore, without any  further restrictions,  the column-wise estimation error cannot be controlled by this estimator. This is  why   the support  of $A^{(s)}$ in (\ref{eq.init-thr}) is also estimated  by thresholding the second  estimator,  $\hat{A}_{j;r}^{(s,YW)}$.  \cite{han2015direct} showed  that, under some conditions which also include Gaussian innovations,  $\|\hat{A}^{(s,YW)}-A\|_1=O_P(\kp \sqrt{\log(p)/n})$, where this rate  refers to the  column-wise estimation error.  Therefore,  
estimator (\ref{eq.init-thr}) allows for the control of both errors made in estimating the coefficients of the parameter matrices $A^{(s)} $, $s=1, \ldots, d$. 

To estimate the innovation covariance matrix $\Sigmaeps$,  several approaches  exist which depend on the sparsity assumptions one wants to impose on  $\Sigmaeps$; see \cite{pourahmadi2013high} for a discussion in the i.i.d. setting. If one  imposes  sparsity assumptions on $\Sigma_\varepsilon$ then  \cite{bickel2008}, \cite{el2008operator}, \cite{cai2011adaptive}, see also  Lemma~\ref{lem.est.var.eps}, provide some thresholding-based  approaches for estimating $\Sigmaeps$. If the  sparsity assumptions are imposed on  $\Sigmaeps^{-1}$,  we refer to   \cite{cai2016estimating},  and if one solely requires that 
 $p/n\rightarrow c \in (0,1)$, to  \cite{,ledoit2012nonlinear} for approaches to  estimate the corresponding covariance  matrix. 
 
 We elaborate here on the case where  sparsity assumptions are imposed on $ \Sigma_\varepsilon$.
Let $\hat \eps_t=X_t-\sum_{s=1}^d \hat A^{(s)} X_{t-j}, t=d+1,\dots,n$, be the estimated residuals 
and assume that $\max_j \|\Sigma_\eps e_j\|_0=O(\cp)$ and  $ \|\Sigma_\eps\|_1=O(\cp)$. For $\bar {\hat {\eps}}_n=1/(n-d) \sum_{t=d+1}^n \hat \eps_t$ let 
$\widetilde \Sigma_{\eps}=(n-d)^{-1}\sum_{t=d+1}^n (\hat \eps_t- \bar {\hat {\eps}}_n)(\hat \eps_t- \bar {\hat \eps}_n)^\top$ and 
\begin{align}
\hat \Sigma_{\eps}=\Big(\widetilde \Sigma_{ \eps, i,j} \ind\{|\widetilde \Sigma_{\eps,i,j}|\geq b_n\}\Big)_{i,j=1,2, \ldots, p}, \label{eq.sigma}
\end{align}
where $\widetilde \Sigma_{\eps,i,j}$ denotes the $(i,j)th$ element of $ \widetilde\Sigma_\eps$; see Lemma~\ref{lem.est.var.eps} for properties of the estimator $\hat\Sigma_\eps$ and more specifically that  this estimator  satisfies Assumption 1(\ref{ass1.4}) of the next section.  In particular, it is shown  in this lemma,  under the assumption of  Gaussian innovations $\eps_t$, that  $\|\hat \Sigma_{\eps}-\Sigmaeps\|_1=O_P(\cp\sqrt{\log(p)/n})$. 
The results of \cite{bickel2008}, Section 2.3 and \cite{el2008operator},  indicate that for the non-Gaussian case  slower rates should  be expected. 
However, \cite{cai2011adaptive} showed that with a more refined thresholding strategy the same rate can also be obtained in the non-Gaussian i.i.d. case. They suggest using  individual thresholding values instead of applying  an universal thresholding parameter.  A simple modification of (\ref{eq.sigma}) taking their considerations into account leads to 
$
\hat \Sigma_{\eps}=\Big(\widetilde \Sigma_{ \eps, i,j} \ind\{|\widetilde R_{\eps,i,j}|\geq b_n\}\Big)_{i,j=1,2, \ldots, p}
$
where $\widetilde R_{\eps,i,j}=\widetilde \Sigma_{\eps,i,j} / (\widetilde \Sigma_{\eps,i,i} \widetilde \Sigma_{\eps,j,j})^{1/2}$.

 Note  that if the estimators $ \hat{A}^{(s)}$ are such that the matrix polynomial $ \hat{A}(z) = I -\sum_{s=1}^d \hat{A}^{(s)} z^s$ has all its roots outside the unit disc, i.e. $\rho(\hat\A)<1$, and if $ \hat\Sigma_\varepsilon $ is positive definite, then $ \hat\Gamma^{(st)}(0)$ is positive definite and the estimator   $\hat{\beta}_r  = \Gammah^{-1} e_r/(e_r^\top \! \Gammah^{-1} \! e_r)$  is well defined.  Furthermore, given estimators $ \hat{A}{(s)}$, $s=1,2, \ldots, d$ and $\hat{\Sigma}_\varepsilon$,  an alternative  estimator of  the autocovariance matrix $\Gamma^{(st)}(0)$ of interest can be obtained as 
 \begin{equation}\label{eq.estspec}
 \breve{\Gamma}^{(st)}(0)=  \int_{-\pi}^\pi  \hat{f}(\lambda) d\lambda, \text{ where }   \hat{f}(\lambda) = \big(I_{p} - \sum_{s=1}^d\hat A^{(s)}e^{-i s\lambda}\big)^{-1}\hat\Sigma_{\varepsilon} \big(\big(I_{p} - \sum_{s=1}^d\hat A^{(s)}e^{is \lambda}\big)^{-1}\big)^{\top}, 
 \end{equation}
 $\lambda\in [\pi,\pi],$  is the $p\times p$ spectral density matrix of the estimated VAR(d) process.  In the following we will focus on the estimator $\hat{\Gamma}^{(st)}(0)$ based on expression  (\ref{eq.GammaHat}).

\subsection{Asymptotic distribution of de-sparsified estimators} \label{subsec2.3}
To derive the asymptotic distribution of  $\hat A_{j;r}^{(s,de)}$ we first observe  that
by substituting  expression  $X_{t;j}=e_j^\top \Xi W_{t-1}+\eps_{t;j}$, $\Xi=(A^{(1)},\dots,A^{(d)})$ in (\ref{de_est}), that  this estimator can be  written as, 
\begin{align}    \label{de_est2}
\hat A_{j;r}^{(s,de)} =& A_{j;r}^{(s)}+\big(\sum_{t=d+1}^n \hat Z_{t-1;(s-1)d+r} W_{t-1;(s-1)d+r}\big)^{-1}
\Big[\big(\sum_{t=d+1}^n \hat Z_{t-1;(s-1)d+r} \eps_{t;j}\big) \\
& \ \ \ \ +
(\sum_{t=d+1}^n \hat Z_{t-1;(s-1)d+r} W_{t-1;-((s-1)d+r)}(\Xi_{j;-((s-1)d+r)}-\hat \Xi_{j;-((s-1)d+r)}^{(re)})\Big]. \nonumber
\end{align}
Representation $(\ref{de_est2})$ suggests that  asymptotic normality of the de-sparsified  estimator, more precisely  of $ \sqrt{n}(\hat A_{j;r}^{(s,de)}-
  A_{j;r}^{(s)} )$ 
can be established,   due  to the contribution of the second term on the right hand side of (\ref{de_est2})  and the asymptotic negligibility 
of the last term  
 since this term depends mainly on   the estimation error  $ \Xi_{j;-((s-1)d+r)}-\hat \Xi_{j;-((s-1)d+r)}^{(re)}$. 
Theorem \ref{thm.clt} below confirms that this intuition is indeed  true. However, to state  precisely this theorem, we impose some conditions   on the underlying VAR process,  on its sparsity as well as on  
    the consistency properties of the   estimators involved in the construction of the  de-sparsified estimator $\hat A_{j;r}^{(s,de)}$. 

\begin{asp} \label{ass1} {~}
\begin{enumerate}[(i)]
\item  $\max_{1\leq i \leq p}\sum_{s=1}^d \|e_i^\top A^{(s)}\|_0 = O(\kp)$,  $\max_{1\leq i \leq p}\sum_{s=1}^d \| A^{(s)}e_i\|_0 = O(\kp)$ and
$ \max_{1\leq i \leq p} \| e_i^\top \Sigma_\varepsilon \|_0 = O(\cp)$.
 \label{ass1.1}
\item 
There exists a $\lambda\in(0,1)$ such that $\rho(\A)\leq \lambda$ and  for any $k \in \N$, 
\[\|\mathds{A}^k\|_2 =O(\lambda^k), \|\mathds{A}^k\|_1 =O(\kp \lambda^k) \ \ \mbox{and}\ \ 
\|\mathds{A}^k\|_\infty =O(\kp \lambda^k).\]  
\label{ass1.2}
\item The estimator $ \hat{A}^{(s,re)}$  used in estimating the bias term in (\ref{de_est}) 
satisfies
$$\sum_{s=1}^d \|\hat A^{(s,re)}-A^{(s)}\|_\infty=O_P(\kp \sqrt{\gp/n})$$
while  the estimator $ \hat{A}^{(s)}$ used in estimating $ \hat\Gamma^{(st)}(0)$ satisfies the above  condition regarding the $ \|\cdot \|_\infty$ norm 
and additionally that
\[\sum_{s=1}^d \|\hat A^{(s)}-A^{(s)}\|_1=O_P(\kp \sqrt{\gp/n}). \]
\label{ass1.3}
\item 
$ \|\Sigma_\varepsilon - \hat \Sigma_\varepsilon\|_1 = O_P(\cp \sqrt{\gp/n})$. 
\label{ass1.4}
\item $ \| \Gamma^{(st)}(0)^{-1} \|_\infty = O(\kpp)$.
\label{ass1.5}
\item  $(E (v^\top \eps_0)^q)^{1/q}<\infty$ for all $\|v\|_2<\infty$  for some $q\geq 8$ such that 
$p/(n^{q/4-1}\log^{q/4}(p))=O(1)$ and $\kp^5\cp\kpp^2 \gp/\sqrt{n}=o(1)$.
\label{ass1.6}
\end{enumerate}
Notice that   $g$ appearing in the above expressions  is  an increasing function of the dimension $p$  which is allowed to  be different from expression to expression. 
\end{asp}

Some comments on the above assumptions are in order.  The first two statements  of Assumption 1(\ref{ass1.1}) refer to  the sparsity behavior of the coefficient  matrices $ A^{(s)}$. According to this assumption, $\kp$ is the  maximum number of non-zero coefficients  that can appear in each row, respectively, in each column of  the matrices $ A^{(s)}$. This  implies 
that $ \sum_{s=1}^d\|\operatorname{vec}(A^{(s)})\|_0= O(\kp p)$ which is a much more flexible sparsity setting for VAR(d) models compared, for instance, to the one used in  \cite{basu2015} and which requires that $ \sum_{s=1}^d\|\operatorname{vec}(A^{(s)})\|_0 =k$.  Notice that  increasing the  dimension $p$ of the VAR(d) process  means that new time series are included. If none of these  new time series is a  white noise  processes,  then the  number of non-zero parameters   $\sum_{s=1}^d\|\operatorname{vec}(A^{(s)})\|_0$ will increase by at least the same order as the dimension $p$ of the process increases. Therefore, the requirement that $ \sum_{s=1}^d\|\operatorname{vec}(A^{(s)})\|_0 =k$, together with the assumption of Gaussianity, essentially means   that only i.i.d. processes can be added to the vector $X_t$  if  its  dimension $p$  increases with $n$. The third statement of Assumption 1(\ref{ass1.1}) refers to the sparsity of the innovation covariance matrix $\Sigma_\varepsilon$. They have to be modified appropriately if one prefers to put   sparsity assumption on the inverse matrix $ \Sigma_\varepsilon^{-1}$. 

Assumption~\ref{ass1}$(\ref{ass1.2})$ imply that the VAR model considered is stable and further satisfies some kind of uniformity  regarding the  decay behavior of the  coefficient matrices $\A^k$. This assumption seems    necessary because the dimension $p$ of the process $\{X_,t\in\Z\}$   is allowed to increase to infinity with $n$. To elaborate, let for simplicity $ d=1$. Then it is well known that $ A^k=B_k$, where $ B_k$ are the coefficient matrices appearing in the infinite order causal, moving average representation $ X_t=\sum_{k=1}^\infty B_k \varepsilon_{t-k} + \varepsilon_t$ of $X_t$. Furthermore, these coefficients decrease exponentially fast to zero, that is $ \|A^k\|_{j} \leq C \rho^k$, for $j \in \{1,\infty\}$ and for some constants $C$ and $\rho \in (0,1)$, which will eventually vary when the   dimension of the process changes. This situation  is taking care off in Assumption 1(\ref{ass1.2}) which controls the 
decay rate of $ A^k$ by the same constant $\lambda \in (0,1)$. 

Assumption 1(\ref{ass1.3}) states the required row wise consistency rates for the estimator $ \hat{A}^{(s,re)}$ and the row- and column wise consistency rates of   $\hat{A}^{(s)}$.  Notice first that as this assumption shows, we can  choose $\hat{A}^{(s,re)} $ to be the same estimator as  $\hat{A}^{(s)}$, where  the latter estimator is  given in (\ref{eq.init-thr}). Furthermore, and as we have already mentioned in the discussion following equation (\ref{eq.init-thr}), under Gaussian assumptions on the innovations $\varepsilon_t$, the  estimator $\hat{A}^{(s)}$  
satisfies by construction the required consistency rates with respect to both norms. Observe that we left $g$ unspecified since its  particular form  depends, among other things, also  on the distribution of the i.i.d. innovations. As already mentioned,  for Gaussian innovations, \cite{kock2015oracle} showed under standard lasso conditions that  the corresponding  estimators of $ A^{(s)}$  satisfy  the rates of Assumption~\ref{ass1}$(\ref{ass1.3})$  with the function $\gp$ given by $\gp=(\log(n)\log(p)\log(d))^5$. Similarly results hold also for the sub-Gaussian case, see \cite{zheng2018testing} and \cite{chernozhukov2018lasso}. If the distribution of  the i.i.d innovations possesses has only a limited number of finite moments, \cite{chernozhukov2018lasso} pointed out that $g$ drops to a polynomial rate; see among others Comment 5.4 of the aforementioned paper.

Assumption 1(\ref{ass1.4}) refers to the consistency rates of the estimator $\hat\Sigma_\eps$. Notice that as Lemma~\ref{lem.est.var.eps} shows, the thresholded estimator $ \hat\Sigma_\eps$ given in (\ref{eq.sigma}) satisfies under the imposed sparsity assumptions and for  Gaussian innovations the  consistency condition stated in this assumption. Assumption 1(\ref{ass1.5}) controls the grow rates of inverse of the stacked covariance matrix $\Gamma^{(st)}(0)$. Note that the stacked covariance itself is affected by the increasing number of non-zero elements $\kp$ and $\cp$ but the inverse $ \Gamma^{(st)}(0)^{-1}$ is also affected by the smallest eigenvalue of $\Gamma^{(st)}(0)$. Assumption 1(\ref{ass1.6}) imposes a moment condition on the innovations $\varepsilon_t$ and specifies its effect on the allowed growth rates of the sparsity behavior of the VAR process when the dimension $p$ of the process  increases with $n$.

We now  state a  result which establishes asymptotic normality of the de-sparsified estimators given in (\ref{de_est}).

\begin{thm}
 \label{thm.clt} Suppose that Assumption \ref{ass1} is satisfied. Then for all $j,r \in \{1,\dots,p\}$ and $s=1,\dots,d$, we have, as $n\to \infty$, 
\begin{align}
\sqrt n \left(\hat A_{j;r}^{(s,de)}- A_{j;r}^{(s)}\right) &  \overset{D}{\to} \mathcal{N}(0,s.e.(j,r,s)^2),
\end{align}
where
$
s.e.(j,r,s)^2
={\Sigma_{\eps,j;j}}(e_{(s-1)p+r}^\top (\Gammas)^{-1} e_{(s-1)p+r})$.
\end{thm}


The following  lemma   deals with  the limiting covariance of the same  estimators.

\begin{lem}\label{lem.dependence}
Under the conditions of Theorem~\ref{thm.clt} we have  for any  $j_1,j_2,r_1,r_2 \in \{1,\dots,p\}$, and $ s_1,s_2 \in\{1,\dots,d\}$, that as $n \to \infty$, 
\begin{align}
\cov(\sqrt n \hat A_{j_1;r_1}^{(s_1,de)},  \sqrt n \hat A_{j_2;r_2}^{(s_2,de)}) \to \Sigma_{\eps,j_1;j_2} e_{(s_1-1)p+r_1}^\top (\Gammas)^{-1}e_{(s_2-1)p+r_2}.
\end{align}
\end{lem}

\section{Bootstrap Based Inference for  Sparse Vector Autoregressions}

\subsection{Bootstrapping de-sparsified estimators}  \label{se.bootalgo}
Due  to their regular limiting distribution,  de-sparsified estimators   can be used as a  vehicle for  statistical inference for   sparse VAR($d$)  models. 
In this section we introduce a  bootstrap procedure for this purpose. Note that in contrast to the i.i.d. regression setting, in the VAR($d$) setting considered in this paper,  the ``regressors'' $(X_{t-1},\dots,X_{t-d},t=d+1,\dots,n)$, i.e. the vectors  $\{W_{t-1}, t=d+1, \ldots, n\}$  are not fixed; they are dependent random vectors. In order to capture the randomness and the dependence properties of these vectors,  an appropriate  bootstrap procedure has  to be applied,  which uses an estimated and   sparsified version of the underlying  VAR$(d)$ model
to  generate  pseudo time series  $X_1^*, X_2^*, \ldots, X_n^*$.  
This pseudo time series can then be used  to infer properties of $ A_{j;r}^{(s)}$  using the estimated distribution of     $\sqrt n (\hat A_{j;r}^{(s,de)}- A_{j;r}^{(s)}) $. 

A standard procedure to bootstrap a VAR process in the finite, low dimensional case, 
is 
to estimate the model and to generate pseudo time series by using the estimated model structure and by drawing with replacement from the estimated set of residuals. However, this procedure  fails in the high-dimensional case. To elaborate, let $\hat \eps_t$ be the centered, estimated residuals of a VAR fit. If the pseudo innovations $\eps_t^*$ are obtained by drawing with replacement from the set $\{\hat \eps_t, t=d+1,\dots,n\}$, then $\var^* \eps_t^*=1/(n-d)\sum_{t=d+1}^n \hat \eps_t \hat \eps_t^\top=\tilde \Sigma_\eps.$ However,  it is well known that  the sample covariance matrix $\tilde \Sigma_\eps$ is not a  valid estimator  of $ \Sigma_\varepsilon$  in the high-dimensional case. Thus, such a bootstrap approach would fail to appropriately mimic the second-order properties of the time series
at hand. To tackle this problem, a different strategy is proposed.
To elaborate note  first  that as Theorem \ref{thm.clt} shows, the limiting distribution of the de-sparsified estimators is not affected by  the  distribution of the  $\varepsilon_t$'s but solely by their  second order properties, that is by the covariance matrix $ \Sigma_\varepsilon$. Thus, we generate the pseudo innovations  $\{\varepsilon_t^*\}$ as Gaussian random vectors  with covariance matrix $\Sigmah$. The estimator $\Sigmah$ is consistent in the high-dimensional setting which ensures that $\{\varepsilon_t^*\}$  generated in this way will  correctly mimic  the second-order properties of the true innovations $\{\eps_t\}$. Notice that instead of the Gaussian distribution any other distribution could be used as well,  provided this distribution  possesses at least $q$ moments. 
Furthermore,  using the Gaussian distribution to generate the pseudo innovations, does not affect the  asymptotic validity of the bootstrap procedure proposed in the sequel, when this procedure is applied to  estimate the  distribution of the de-sparsified estimators and to perform tests for groups of coefficients.


Apart from using the Gaussian distribution with covariance matrix $ \hat{\Sigma}_{\eps}$ to generate the  pseudo innovations, 
the thresholded estimators 
$\hat A^{(s)}$ given in (\ref{eq.init-thr})  
are used to generate the pseudo time series. This  enables  the bootstrap  generated 
 pseudo time series $X_1^*, X_2^*, \ldots, X_n^*$ to  appropriately   mimic   the  dependence and the sparsity  properties of the underlying VAR($d$) model. Note that 
thresholding is important in the sparse  high dimensional context since it guarantees sparsity of the VAR model used in the bootstrap world.  Using the aforementioned  estimators $ \hat{A}{(s)}$ and $ \hat{\Sigma}_\eps$, the generated pseudo time series $X_1^*, X_2^*, \ldots, X_n^*$   stems from a VAR($d$) model  which has $\Gammah$ as its lag-zero autocovariance matrix.
Recall that the estimator $\hat{A}^{(s)}$ given  in (\ref{eq.init-thr}),   selects the non-zero components  of the matrices $ A^{(s)}$ by means of thresholding  two  initial estimators.  Notice that, as we will see,  consistency of this estimator  as well as 
 a bounded row- and column-wise support property is sufficient for  validity of the  bootstrap procedure  proposed.  More specifically, and as a careful inspection of the proof of Theorem~\ref{thm.bootstrap.lasso}  below shows, a full sign recovery, which will require a minimal signal strength condition on the coefficients $A^{(s)}_{j;r} $, is not needed. Notice that  under such an additional minimal signal strength condition,  consistent estimation of the support of $A^{(s)}$, $s=1,\dots,d$, also  can  be established; see Theorem 7 of \cite{kock2015oracle} and Corollary 1 of \cite{han2015direct}.  Lack of such a    minimal signal strength condition 
 does not affect the consistency and the  aforementioned row- and columnwise boundedness properties of the  estimators $ \hat{A}^{(s)} $,  which are used for the bootstrap.  


The  bootstrap algorithm  proposed   consists   of  the following four steps. 

\begin{enumerate}
\item[] {\it Step 1:} \  Generate i.i.d. pseudo innovations $\eps_t^*$  from $\mathcal{N}(0,\hat\Sigma_\eps)$. 
\item[] {\it Step 2:} \  Generate a pseudo time series  $ X_1^*, X_2^*, \dots,X_n^* $  using the model equation 
$$X_t^*=\sum_{s=1}^d \hat A^{(s)} X^*_{t-s} + \eps_t^*, \ \ t=1,2, \ldots, n$$
and some starting values $ X_0^*, X_{-1}^*, \ldots, X^*_{1-d}$.
 \item[] {\it Step 3:} \ Let $ \hat A_{j;r}^{*(s,de)}$ be the same de-sparsified estimator of $ A_{j;r}^{(s)}$ as  the estimator $ \hat A_{j;r}^{(s,de)}$ given in $(\ref{de_est})$,  but based on the pseudo time series $X_1^*, X_2^*, \ldots, X_n^*$. 
 \item[] {\it Step 4:} \ Approximate the distribution of $  \sqrt n (\hat A_{j;r}^{(s,de)}- A_{j;r}^{(s)})$ by that of the bootstrap analogue $\sqrt n (\hat A_{j;r}^{*(s,de)} - \hat A_{j;r}^{(s)}) $.
\end{enumerate}

Validity of the bootstrap procedure in approximating consistently the distribution of de-sparsified estimators  is established in the following theorem, where  Mallow's $d_2$ metric is used to measure the distance between two distributions.  
 For two random variables $X$ and $Y$,    Mallow's distance between their distributions  is  defined as 
 $d_2(X,Y)=\{ \int_0^1 \left(F_X^{-1} (x)-F_Y^{-1}(x)\right)^2 dx\}^{1/2}$ where   $F_X$ and $F_Y$ denote the cumulative distribution functions of  $X$ and $Y$, respectively; see \cite{bickel1981}. 

\begin{thm}
\label{thm.bootstrap.lasso}
Suppose that Assumption 1 is fulfilled.  Then, we have  for all $j,r \in \{1,\dots,p\}$ and  $ s\in\{1,\dots,d\}$, that as $n\to \infty$, 
\begin{align}
d_2\left(\sqrt n (\hat A_{j;r}^{(s,de)}- A_{j;r}^{(s)}),\sqrt n (\hat A_{j;r}^{*(s,de)} - \hat A_{j;r}^{(s)})\right)=o_P(1).
\end{align}
\end{thm}

\vspace*{0.3cm}

The next corollary establishes validity  of the corresponding studentized distributions as well. 

\begin{corollary}
\label{co.stud} Under the assumptions of Theorem~\ref{thm.bootstrap.lasso} we have, as $ n\rightarrow \infty$, that 
\begin{align}
d_2\left(\sqrt n (\hat A_{j;r}^{(s,de)}- A_{j;r}^{(s)})/ (\widehat {s.e.}(j,r,s)),\sqrt n (\hat A_{j;r}^{*(s,de)} - \hat A_{j;r}^{(s)})/(\widehat {s.e.}^*(j,r,s))\right)=o_P(1),
\end{align}
where 
\begin{align}
\widehat {s.e.}(j,r,s)^2= \hat \Sigma_{\eps,j;j} \left( e_{(s-1)p+r}^\top (\Gammah)^{-1} e_{(s-1)p+r}\right) 
\end{align}
and
\begin{align}
\widehat {s.e.}^*(j,r,s)^2=\hat \Sigma_{\eps,j;j}^* \left((e_{(s-1)p+r}^\top (\Gammahast)^{-1} e_{(s-1)p+r} \right).
\end{align}
Here $\hat \Sigma_{\eps,j;j}^*$ and  $\Gammahast$  denote the same quantities as $\hat \Sigma_{\eps,j;j}$ and $\Gammah$  but based on the bootstrap  pseudo time series   $X_1^\ast, X_2^\ast, \ldots, X_n^\ast$. The estimator $\hat \Sigma_{\eps,j;j}$ is the $j$th diagonal entry of the estimator $\hat\Sigma_\eps$ used in $ Step\ 1$ of  the bootstrap algorithm.
\end{corollary}
The last result of this section shows that the bootstrap version of the de-sparsified estimators considered, also can be  used to consistently estimate the 
 covariance of  the de-sparsified estimators of the coefficients $ A_{j_1;r_1}^{(s_1)} $ and  $ A_{j_2;r_2}^{(s_2)}$. Its proof follows along the same lines as the proof of  
 Lemma \ref{lem.dependence} and uses  arguments similar to those applied in the proof of Theorem \ref{thm.bootstrap.lasso}.

\begin{lem}\label{lem.dependence2}
Under the assumptions of Theorem~\ref{thm.bootstrap.lasso} we have  for any $j_1,j_2,r_1,r_2 \in \{1,\dots,p\}$ and $ s_1,s_2 \in\{1,\dots,d\}$,  that as $n \to \infty$, 
\begin{align}
\cov^\ast(& \sqrt n  \hat A_{j_1;r_1}^{*(s_1,de)},  \sqrt n \hat A_{j_2;r_2}^{*(s_2,de)}) 
\to \Sigma_{\eps,j_1;j_2} e_{(s_1-1)p+r_1}^\top (\Gammas)^{-1} e_{(s_2-1)p+r_2},
\end{align}
 in probability.
\end{lem}

\subsection{Testing statistical hypotheses} \label{sec.boottest}
Appropriately modified, the bootstrap procedure proposed also can  be used  in order to test   hypotheses of interest regarding  the dependence structure of the underlying VAR($d$)  model, like for instance,   testing that a subset of the parameters  of the VAR model is zero. To elaborate, let   $ G \subseteq \{(j,r,s) : j,r \in \{1,\dots,p\}  \ \mbox{and}  \ s \in\{1,\dots,d\} \}$
 be  a  subset of indices and consider  the following testing problem:

\begin{enumerate}
\item[$H_0$:] $ A_{j;r}^{(s)}=0$, for all  $(j,r,s) \in G$.
\item[$H_1$:]  There exists at least one $(j,r,s) \in G$ such that $A_{j;r}^{(s)}\not = 0$.
\end{enumerate}

We assume  in the following  that the restrictions on the parameter space imposed by the above null and alternative hypotheses, 
do  not violate the  causality  of the underlying VAR($d$) model, 
that is, we assume that the condition  $\det(\mathcal{A}(z))\not =0$ for all $z \leq  1$, is satisfied  under $H_0$ and under $H_1$.  

In order to test the above hypotheses and to avoid problems associated with inverting large scale covariance matrices, we propose to use the max-type  test statistic
\begin{align}
T_n=\max_{(j,r,s) \in G} \left\{\frac{\sqrt{n}|\hat A_{j;r}^{(s,de)}|}{\widehat {s.e.} (j,r,s)} \right\}.
\end{align}
For $\alpha \in (0,1) $ small, let $ m_{n,\alpha} $ be the upper $\alpha$-quantile of the distribution of $ T_n$ under the assumption that the null hypothesis is true. 
 $H_0$ is then rejected if $ T_n>  m_{n,\alpha}$. Critical values of  the test $T_n$ can be obtained using the model-based bootstrap procedure proposed in this paper. For this, the estimated and sparsified VAR model  used to generate the pseudo time series $X_1^*, X_2^*, \ldots, X_n^*$ is   modified in such a way   that it  
 satisfies the null hypothesis. This is important for  a  good size and  power behavior of the bootstrap based test.  To achieve this goal, the 
 parameter matrices $A^{(s)}$, $s=1,2, \ldots, d$,  are estimated under the restrictions imposed by the   null hypothesis, that  is   under the constrains $\hat A_{j;r}^{(s)}=0 $ for all  indices $( j,r,s) \in G$.  Using these restrictions on the matrices $\hat{A}^{(s)}$, the bootstrap algorithm proposed in Section~\ref{se.bootalgo}  is then applied to generate pseudo time series $ X_1^*, X_2^*, \ldots, X_n^*$ from which the bootstrap analogue of $T_n$ under validity of $H_0$   is calculated and which is given by  
 \begin{align}
T^*_n=\max_{(j,r,s) \in G} \left\{\frac{\sqrt{n}|\hat A_{j;r}^{*(s,de)}|}{\widehat {s.e.}^* (j,r,s)}\right\}.
\end{align} 
Let $ m^*_{n,\alpha} $ be the upper $\alpha$-quantile of the distribution of $ T^*_n$.  The bootstrap based test proceeds then by rejecting  $H_0$  if $ T_n>  m^*_{n,\alpha}$.
As the following theorem shows, the bootstrap succeeds  in consistently estimating the distribution of $ T_n$ under the null,  justifying, therefore, the use  of the bootstrap  critical values $ m^*_{n,\alpha}=F^{-1}_{T_n^*}(1-\alpha)$   for performing the test.

\begin{thm}
\label{thm.bootstrap.test}
If Assumption 1 under validity of the null hypothesis is fulfilled 
 then, as $n\to \infty$, 
\begin{align}
\sup_{x\in \R}\Big| P_{H_0}\big(T_n\leq x\big) - P\big(T^*_n \leq x| X_1, X_2, \ldots, X_n\big)\Big| =o_P(1),
\end{align}
where $ P_{H_0}(T_n\leq \cdot)$ denotes the distribution function of $T_n$ when  $ H_0$ is true.
\end{thm}

We can extent the applicability of the bootstrap procedure proposed also 
to the case where  the set $G$  of null hypotheses considered increases to infinity,  at an appropriate rate,   as $n$ increases to infinity. To state this dependence on the sample size, we write $G_n$ for $G$.
We then have the following theorem.

 \begin{thm} \label{thm.test.grow}
Let $|{\Gn}|=O(n^b), b>0$ such that $n^{1/q-1/2+b/q}((b+1)\log(n))^{3/2}=o(1)$. Then under the assumptions of Theorem~\ref{thm.bootstrap.test} we have that, as $n \rightarrow \infty$, 
$$
\sup_{c \in \R} \Big| P^\ast( \max_{(j,r,s) \in \Gn} \big| \frac{\sqrt{n}|\hat A_{j;r}^{*(s,de)}|}{\widehat {s.e.}^* (j,r,s)}\big| \leq c )-
P_{H_0}( \max_{(j,r,s) \in \Gn} \big| \frac{\sqrt{n}|\hat A_{j;r}^{(s,de)}|}{\widehat {s.e.} (j,r,s)}\big| \leq c )\Big| =o_P(1).
$$
\end{thm}

\section{Numerical Results}
We investigate the finite sample performance of the bootstrap procedure proposed  to infer properties of the sparse VAR model by means of simulations and we also discuss a real-life data application.  Notice that  we can use our bootstrap procedure to   construct  confidence intervals for the coefficients of the VAR model as well as to perform  tests of statistical hypotheses about the VAR parameters.   In this section, we focus on the problem of testing hypotheses about groups of model parameters. 
  
  All results presented in this section  are based on implementations in  \emph{R}, \citep{R}, of the procedures proposed in this paper. In the simulations as well as in the real-life data example, the estimator $\hat{A}_{j;r}^{(s)} $ of $A^{(s)}_{j:r}$ is based on the adaptive lasso; see Section 4 in  \cite{kock2015oracle}. To simplify calculations we do not make use of the second estimator $ \hat A^{(s,YW)}_{j;r}$ appearing in (\ref{eq.init-thr}). The reason for this is that  the adaptive lasso estimates showed a  very good finite sample performance regarding both norms, that is the $\|\cdot\|_1$ and the $\|\cdot\|_\infty$ norm and no   tendency towards an 
  out of scale increasing of the column-wise support has been observed. Moreover, in preliminary simulations we found that the adaptive lasso estimator 
  was not outperformed by the combined estimator (\ref{eq.init-thr}) and    the latter estimator  was computationally much more demanding.  
  Thus the estimator  $\hat{A}_{j;r}^{(s)}$ used in our simulations, was  obtained by thresholding  the adaptive lasso estimator with tuning parameter $\lambda_n$ and threshold parameter $a_n=\lambda_n$. The same estimator was  used    for the construction of  the de-sparsified estimators in the  bootstrap procedure.  
  The Bayesian Information Criterion (BIC) has been  used to select the tuning parameter $\lambda_n$. More specifically, the adaptive lasso implementation with BIC selection of $\lambda_n$   of the \emph{HDeconometrics} package,  \cite{garcia2017real},  which uses the \emph{glmnet} package,  \cite{glmnet}, has been applied.  Finally, the covariance matrix $\Sigma_\eps$  of the innovations has been  estimated using (\ref{eq.sigma}),
  where the threshold parameter $b_n$ has been chosen  by cross-validation and  using the implementation of \cite{cov_cross}.

In our numerical investigations, we have experienced that 
in situations where the set of hull hypotheses is  considerably large, that is,  $|G|$ is of the same order as $n$, the bootstrap-based test and  despite its consistency,  seems to be somehow conservative. This  is due to a finite sample bias appearing in estimating the upper percentage point  of the distribution of $T_n$ under $H_0$.  To improve the finite sample behavior the bootstrap based test, and inspired by \cite{efron1981nonparametric} for the construction of confidence intervals, see also \cite{efron1994introduction}, Section 14.3, 
we adapt to the testing context a bias correction procedure. This procedure is described in more detail in Appendix B of the Supplement Material and leads to the following bias-corrected bootstrap percentage point
$m^{*(B)}_{n,\alpha}=F^{-1}_{T^*_n} \Big( \Phi\big(\sqrt{2}\cdot \Phi^{-1}\big(P(T^+_n\leq T_n^*)\big) +  z_{\alpha}\big)\Big) $,
where $ \Phi$ is the distribution function of the standard normal and $\Phi(z_\alpha) =1-\alpha$.
Notice that estimation of the quantity $P(T^+_n\leq T_n^*) $ used in this bias correction procedure requires the implementation of a second bootstrap
experiment; see Appendix B for more details.

\subsection{Simulations} \label{sec.sim}
{\bf Example 1:} We generated  time series $X_1,X_2,\dots,X_n$ from the  VAR(1) model,
$X_t=A X_{t-1}+\eps_t$, with i.i.d. $  \eps_t \sim \mathcal{N}(0,\Sigma_\eps)$,  where 
 the coefficient matrices $A$ and $\Sigma_\eps$ possess a cluster (or  block) structure. That is,  only the coefficients within a small number of clusters are allowed to 
 be different from zero; see also 
 \cite{han2015direct} for the use of such VAR models. Each block is of size $20 \times 20$ and is given by the matrices $A_{BLOCK}^\xi$ and $\Sigma_{BLOCK}$ in Appendix C. Four entries on the main diagonal of $A_{BLOCK}^\xi$ are specified  by the choice of the  parameter $\xi$. This parameter controls in some sense the level of  dependence of the generated VAR model obtained. Two values  $\xi =0.6 $ and $ \xi=0.9$  are considered, which lead to the  maximal absolute eigenvalues  $\lambda_{\max}=0.7$ and $\lambda_{\max}=0.9 $, respectively,  of the coefficient matrix $A$.
 Note that the information that $A$ possesses a cluster structure was not used in the subsequent inference procedure. Therefore,  the same results could be obtained for the case where the indices are shuffled randomly; see Figure~\ref{fig:Acluster} for an illustration of the matrix $A$ and of a randomly shuffled version. The dimension $p$  of the VAR process is set equal to  $20, 100$ and $200$ where the case   $p=20$ results in a single  cluster. 


\begin{figure}
    \centering
    \includegraphics[width=0.495\textwidth]{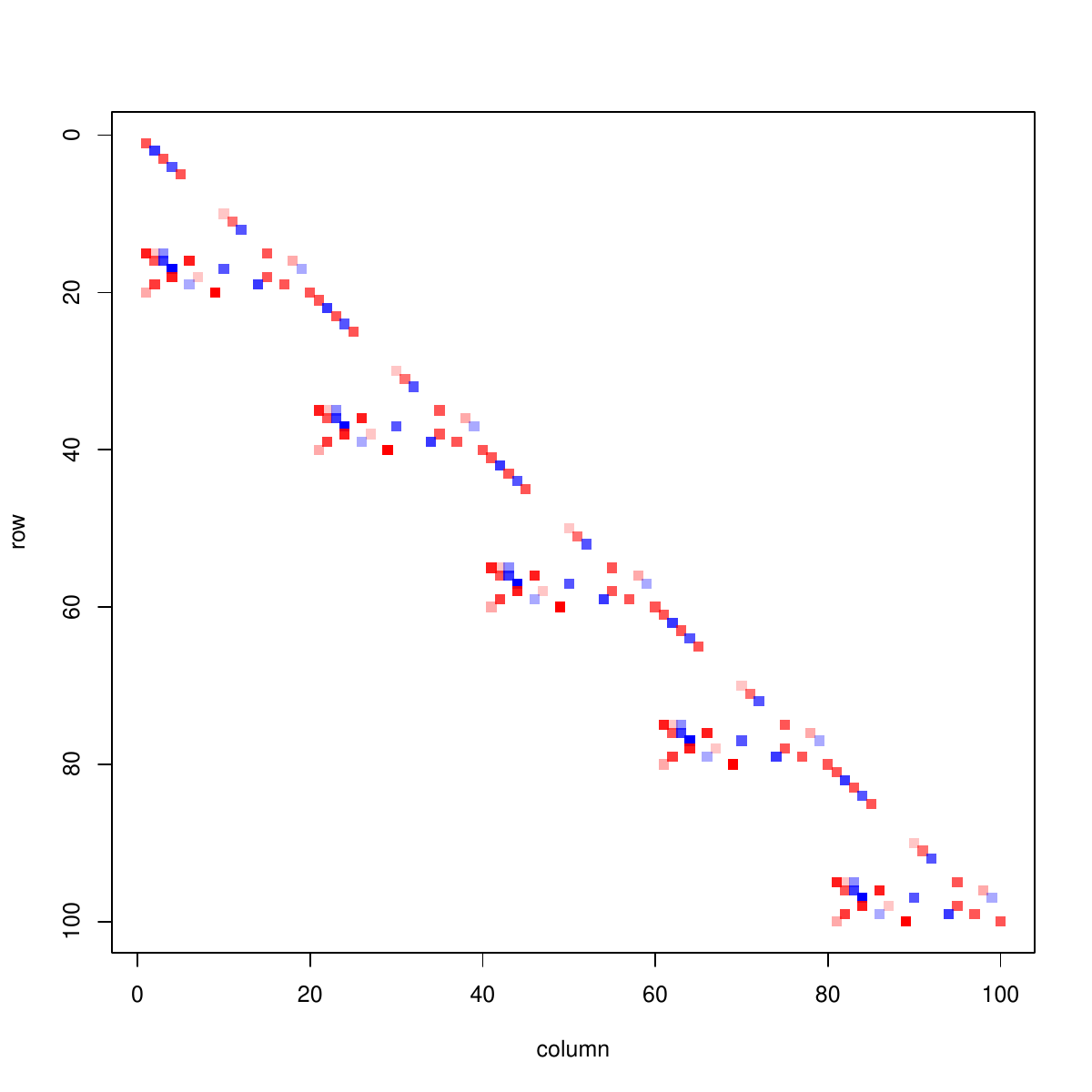}
    \includegraphics[width=0.495\textwidth]{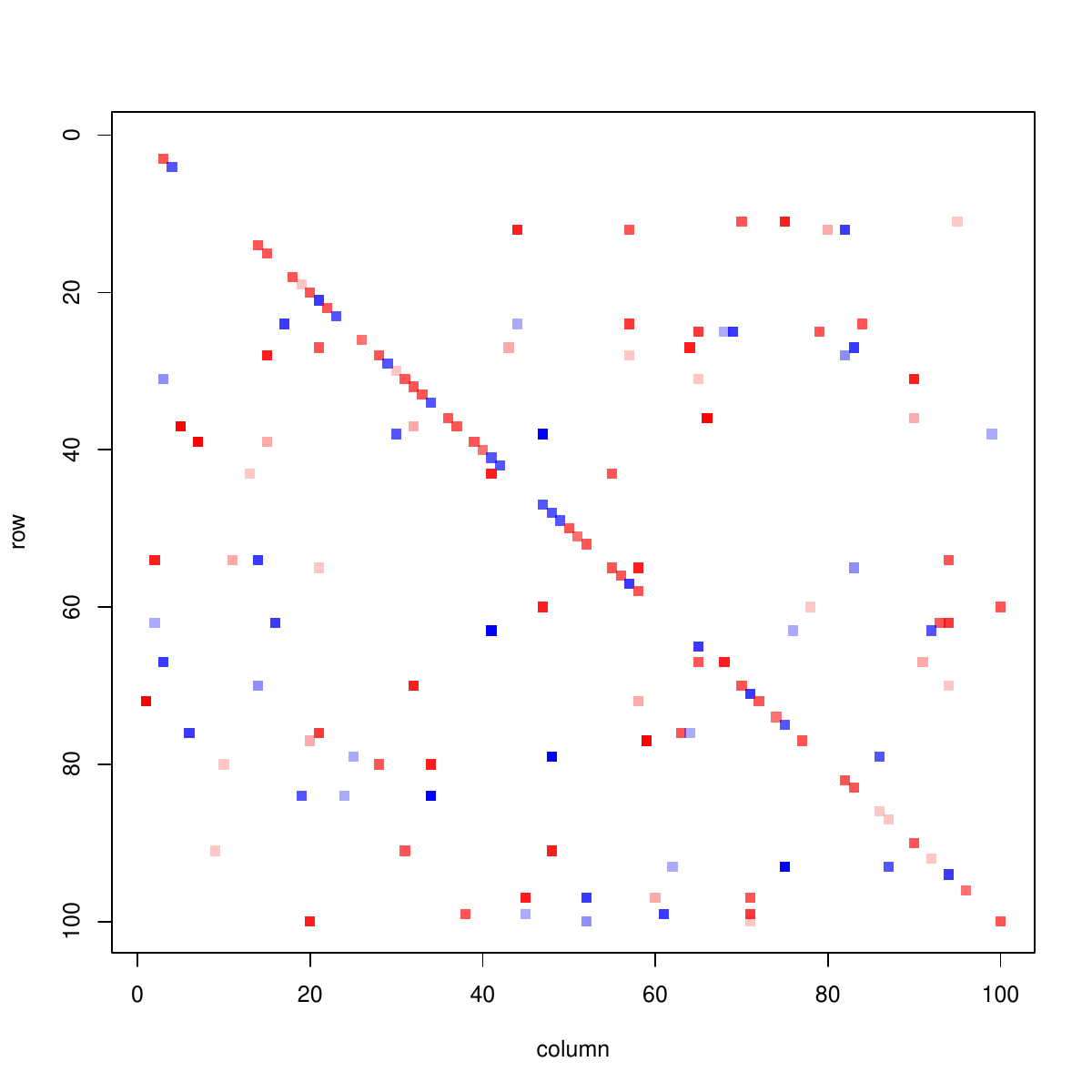}
    \caption{Structure of the coefficient matrix  $A$  for $p=100$: The left panel shows the non-zero elements of $A$ with ordered indices  and 
    the right panel the non-zero elements of the same matrix with shuffled indices. Positive coefficients are presented by red  and negative coefficients by blue dots.}
    \label{fig:Acluster}
\end{figure}

Let $G_1=\{(i,j) : i \in \{1,\dots,20\}, j =\{p-20+1,\dots,p\}$ be the first  set  
and   $G_2=\{(i,j) : i \in \{1,\dots,10\}, j =\{p-10+1,\dots,p\}$ be a second set of indices which correspond  to the upper-right corner of the matrix $A$. Note that $|G_1|=400$ while  $|G_2|=100$.
We  consider for $k \in \{1,2\}$, the  testing problem,
\begin{enumerate}
\item[$H_0$:]  $ A_{j;r}=0$  for all  $(j,r) \in G_k$, against
\item[$H_1$:]  There exists $(j,r) \in G_k$ such that  $ A_{j;r}\not = 0$.
\end{enumerate}
For $k=1$, $H_0$ corresponds to the case that the time series belonging to  the first cluster, that is,  $X_{t,1},\dots,X_{t,20},$ are not directly influenced  by the set of lagged time series  $X_{t-1,p-20+1},\dots,X_{t,p}$ belonging to  the last cluster. A similar interpretation occurs  for $k=2$ in which $H_0$ corresponds to the case where the first ten time series are not directly influenced by the  last ten lagged time series.
For  $p=20$ only the case  $k=2$ is considered since for $k=1$ the corresponding null hypothesis is that $X_t$ is a white noise process. We investigate the performance of the bootstrap based test under the null as well as its power against various  alternatives.  The alternatives considered  refer to the case where only one coefficient in the set $G$ is set different to zero and  equal to $\delta_a=0.3$.  Notice that $ \delta_a=0$ corresponds to the null hypothesis. The results obtained for $500$ repetitions  and $B=1,000$ bootstrap replications  are presented  in Table~\ref{example_H_0.table} (size) and  Table~\ref{example_H_1.table} (power). 

As Table~\ref{example_H_0.table} shows, the bootstrap-based test without the proposed bias correction, seems to be   conservative in most of the cases
and the differences between the empirical and the nominal level increase as  the size of the set $G$ increases. The degree of persistence as well as the dimension of the process seem not  to affect the size behavior  of the test. However, using the bias correction, the behavior of the bootstrap-based test improves considerably for all  sizes of the  set $G$ considered. 
Table~\ref{example_H_1.table} presents the empirical power of the bootstrap based test. As it can be seen, in all 
settings considered, the bias correction improves the  power behavior  of the test, where (as expected) 
the test is more powerful  for the  set $G_2$ than for the more larger set $G_1$. Notice that for the sample size of $n=128$ observations, detecting the deviation from the null  which is due to
the fact that only one out of the $ |G_1|=400$ elements is set equal to $0.3$,  is a very  challenging inference problem.  However, even in this case, the test has power and its power improves considerably as the 
 sample size increases.
Finally, as in the case of the size,  the dimension of the VAR process seems  to affect only slightly the power of the test. 

\begin{table}[h]
\begin{tabular}{|l|l|rrrrrr||rrrr|} 
\hline
                     &            & \multicolumn{6}{c||}{$G_2$ }                                                 & \multicolumn{4}{c|}{$G_1$ }                         \\ 
\hline
                     &      $p$      & \multicolumn{2}{c}{20} & \multicolumn{2}{c}{100} & \multicolumn{2}{c||}{200} & \multicolumn{2}{c}{100} & \multicolumn{2}{c|}{200}  \\ 
\hline

                   & $\alpha$  & 0.05 & 0.10            & 0.05 & 0.10             & 0.05 & 0.10               & 0.05 & 0.10             & 0.05 & 0.10               \\ 
\hline
$n$ &  $\lambda$& \multicolumn{10}{|l|}{With Bias Correction}\\

\hline
\multirow{2}{*}{128} & 0.7 & 0.04 & 0.08 & 0.05 & 0.11 & 0.03 & 0.07 & 0.05 & 0.11 & 0.03 & 0.06\\

&0.9 & 0.04 & 0.10 & 0.05 & 0.10 & 0.05 & 0.10 & 0.04 & 0.10 & 0.05 & 0.09\\

\multirow{2}{*}{200} & 0.7 & 0.04 & 0.10 & 0.05 & 0.11 & 0.06 & 0.09 & 0.04 & 0.10 & 0.06 & 0.11\\

&0.9 & 0.06 & 0.12 & 0.05 & 0.10 & 0.05 & 0.12 & 0.06 & 0.12 & 0.04 & 0.09\\
\hline
$n$& $\lambda$ &\multicolumn{10}{|l|}{Without Bias Correction}\\
\hline

\multirow{2}{*}{128} & 0.7 & 0.03 & 0.06 & 0.03 & 0.08 & 0.02 & 0.05 & 0.03 & 0.06 & 0.01 & 0.04\\

&0.9 & 0.04 & 0.07 & 0.04 & 0.08 & 0.03 & 0.07 & 0.02 & 0.05 & 0.03 & 0.06\\

\multirow{2}{*}{200} & 0.7 & 0.04 & 0.08 & 0.04 & 0.10 & 0.05 & 0.08 & 0.02 & 0.07 & 0.04 & 0.08\\

&0.9 & 0.05 & 0.10 & 0.03 & 0.07 & 0.04 & 0.09 & 0.04 & 0.09 & 0.02 & 0.06\\

\hline
\end{tabular}
\caption{Rejection frequencies under  $H_0$ ($\delta_a=0.0$), 
for time series stemming from the VAR model of Example 1 with dimension $p=20, 100, 200$ and   sample sizes  of $n=128$ and $ 200$ observations.} \label{example_H_0.table}
\end{table}

\begin{table}
\begin{tabular}{|l|l|rrrrrr||rrrr|} 
\hline
                     &            & \multicolumn{6}{c||}{$G_2$ }                                                 & \multicolumn{4}{c|}{$G_1$ }                         \\ 
\hline
                     &      $p$      & \multicolumn{2}{c}{20} & \multicolumn{2}{c}{100} & \multicolumn{2}{c||}{200} & \multicolumn{2}{c}{100} & \multicolumn{2}{c|}{200}  \\ 
\hline

                   & $\alpha$  & 0.05 & 0.10            & 0.05 & 0.10             & 0.05 & 0.10               & 0.05 & 0.10             & 0.05 & 0.10               \\ 
\hline
$n$ &  $\lambda$& \multicolumn{10}{|l|}{With Bias Correction}\\
\hline

\multirow{2}{*}{128} & 0.7 & 0.61 & 0.75 & 0.64 & 0.77 & 0.62 & 0.79 & 0.27 & 0.43 & 0.28 & 0.43\\

&0.9 & 0.61 & 0.74 & 0.62 & 0.79 & 0.61 & 0.78 & 0.25 & 0.41 & 0.18 & 0.32\\

\multirow{2}{*}{200} & 0.7 & 0.95 & 0.99 & 0.95 & 0.97 & 0.96 & 0.98 & 0.83 & 0.90 & 0.82 & 0.88\\

&0.9 & 0.95 & 0.99 & 0.96 & 0.98 & 0.96 & 0.98 & 0.82 & 0.89 & 0.76 & 0.88\\

\hline
$n$& $\lambda$ &\multicolumn{10}{|l|}{Without Bias Correction}\\
\hline

\multirow{2}{*}{128} & 0.7 & 0.51 & 0.68 & 0.52 & 0.68 & 0.50 & 0.68 & 0.15 & 0.25 & 0.12 & 0.24\\

&0.9 & 0.51 & 0.68 & 0.52 & 0.70 & 0.52 & 0.69 & 0.12 & 0.24 & 0.08 & 0.17\\

\multirow{2}{*}{200} & 0.7 & 0.93 & 0.97 & 0.94 & 0.97 & 0.94 & 0.96 & 0.73 & 0.85 & 0.73 & 0.82\\

&0.9 & 0.94 & 0.98 & 0.95 & 0.98 & 0.94 & 0.97 & 0.74 & 0.84 & 0.66 & 0.81\\

\hline
\end{tabular}
\caption{Rejection frequencies under  $H_1$ ($\delta_a=0.3$), for time series stemming from the VAR model of Example 1 with dimension $p=20, 100, 200$ and   sample sizes  of $n=128$ and $ 200$ observations.} \label{example_H_1.table}
\end{table}

\FloatBarrier

\subsection{A real-life data example}
In \cite{farmer2015stock} the question has been discussed whether the stock market affects the labor market. In the  aforementioned  paper,  the
 stock market is represented  by the   S\&P 500 index and the labor market by the unemployment rate, that is,  a bivariate time series  is 
  used to investigate the question of interest.   However, restricting the analysis to  a bivariate system  might be problematic since  it implies a loss of information due to the fact that   
   a variety of time series exist which  describe the activities in the two different macroeconomic sectors. The  corresponding  dependence structures might be much more complicated than those  captured 
    by  the particular bivariate   time series considered.  
 Therefore,  and  in order to get a more detailed and deeper inside  into the relations between the labor and the stock market, we consider the Federal Reserve Economic Data (FRED) of  St. Louis Fed's main, which is publicly available at \url{http://research.stlouisfed.org/fred2/}. This data set  contains  $31$ time series describing the activities  in the labor market and $5$ time series describing the stock market.  Furthermore, 
 to  take into account 
  the fact  that other macroeconomic variables may also exist which  
    simultaneously influence the labor and the stock market, all available time series  in the aforementioned  data set are considered,  which refer to a wide range of different economic activities.   The entire    data set  considered contains $p=124$ time series 
 and a complete description  of all  time series used  is given in  Appendix B of this paper.    
  The techniques proposed in  \cite{mccracken2016fred} have been  adopted  to transform this set of  time series to stationary.  Furthermore, to ensure  comparability of the data used in our analysis with the data set used by \cite{farmer2015stock}, the monthly observations  have been  aggregated to quarterly data where the  time period considered begins by  the fourth quarter  of  1979 and ends by the first quarter  of 2011.   
The  number of available observations is then equal to  $n=126$.     As already  mentioned, the question of interest is 
    whether  the  financial sector,  that is the stock market, influences the labor market. Following  \cite{farmer2015stock}, a vector autoregressive model of order $1$ has been used  and  the following   hypotheses
     have been considered:
\begin{enumerate}
\item[$H_0$:]  $ A_{j;r}=0$ for all  $(j,r) \in G$,  
\item[$H_1$:]  There exists $(j,r) \in G$ such that  $ A_{j;r}\not = 0$, 
\end{enumerate}
where $G$ is the set containing the indices corresponding to the two economic sectors, that is 
$G=\{(j,r) : j \in \emph{LABOR} \ \mbox{and} \ r \in \emph{STOCK}\}$. Here $\emph{STOCK}$ denotes  the set of indices referring to the time series of the stock market and $\emph{LABOR}$  to the time series of the labor market. Notice that $|G|=155$. To test the above pair of hypotheses, the test statistic described in Section \ref{sec.boottest} with $B=2,000$ bootstrap replications and bias corrected percentage points have been used. 

 The degree of  sparsity obtained  depends on the choice of the regularization  parameters. As mentioned in the previous section, we use the adaptive lasso with regularization parameter $\lambda_n$ and threshold parameter  $a_n=\lambda_n$ to estimate the coefficient matrix, while  the covariance of the innovations $\Sigmaeps$ is estimated  as  in (\ref{eq.sigma}), 
 with  the threshold parameter  $b_n$  chosen by cross-validation. For instance,    for $\lambda_n=0.25$ we obtain a  sparsity in the coefficient matrix $A$ of $0.56\%$,   for 
 $\lambda_n=0.1$  the  sparsity obtained is  $2.7\%$ and if $\lambda_n$ is chosen by BIC, it  leads 
 to a sparsity level of $1.43\%$. 
 The p-values  of the test   seem not to be largely affected  by the choice of this regularization parameter. This 
 is demonstrated in Figure \ref{fig_real_1} where the $p$-values of the bootstrap based  test proposed,  are shown for different values of the regularization parameter $\lambda_n$.
 

\begin{figure}[ht]
\includegraphics[width=\textwidth]{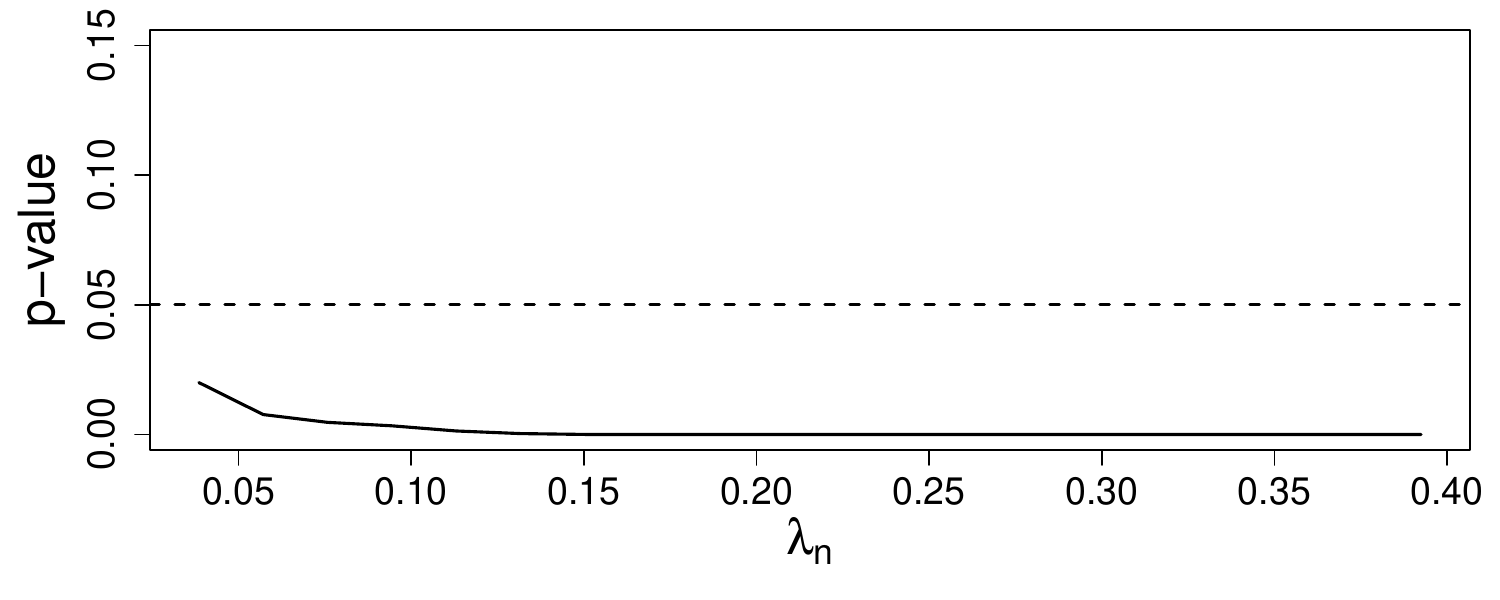}
\caption{P-values (vertical axis) of the test $T_n$  
described in Section \ref{sec.boottest} applied to the FRED data set  for  several  values of the regularization parameter $\lambda_n$ (horizontal axis).} \label{fig_real_1}
\end{figure}
 As it is  seen  from  this  figure, the null hypothesis of interest is rejected at the commonly used $ \alpha=0.05$ level, for all 
   values  of $\lambda_n$ considered. If $\lambda_n$ is chosen by BIC, the test gives a p-values of less than $0.0002$. Furthermore, our bootstrap procedure identifies the following  coefficients $A_{j;r}$, for $ j \in \emph{LABOR}$ and $r \in \emph{STOCK}$, as different from zero: 
  $A_{\emph{USGOOD},\emph{S.P.500}}=0.13$, $A_{\emph{MANEMP},\emph{S.P.500}}=0.18$, $A_{\emph{CLAIMSx},\emph{S.P.div.yield}}=0.32$, and $A_{\emph{NDMANEMP},\emph{S.P..indust}}=0.16$. 
Thus, our analysis 
leads  to the conclusion that the hypothesis that the stock market  does not influence the labor market should  be rejected. Notice that this    conclusion 
is also supported  by the findings reported in    \cite{phelps1999behind} and \cite{farmer2015stock}.

\section{Conclusions}

In this paper,  high-dimensional and  sparse vector autoregressive models have been considered.  We have first  
adopted the concept of de-biasing  to the high-dimensional  VAR($d$) time series  context and 
have considered  de-biased respectively  de-sparsified estimators of the autoregressive parameters. The 
 asymptotic distribution of the de-sparsified  estimators has  been derived under  general conditions. Furthermore, a bootstrap procedure has been proposed which 
 is asymptotically able  to generate pseudo time series that appropriately  imitate the dependence structure and  the sparsity properties of the underlying   high-dimensional  VAR($d$) process. Asymptotic validity of the 
 bootstrap procedure proposed for estimating the distribution of  de-sparsified estimators has been established. Furthermore, an appropriately modified version of the bootstrap  procedure 
  has been  used  for  testing hypotheses about groups of model parameters.  We allow for these groups of model parameters to increase to infinity with sample size. Validity of the bootstrap procedure 
  also has  been  established for this case. Finally, 
   we have demonstrated by means of  numerical investigations,
 the  good  finite sample behavior   of the bootstrap-based  inference procedure proposed and we have analyzed an  interesting  real-life data set.   

\section{Proofs}\label{sec.proofs}
Proofs which are note given in this section are 
presented in Appendix E of the Supplementary Material.
\begin{lem} \label{lem.physical}
Let $Y_{t;i}=G_i(\eps_t,\eps_{t_1},\dots,), i=1,\dots,\tilde p, t \in \Z,$ be some process generated causally by the i.i.d. processes $\{\eps_t\}$ for some function $G=(G_1,\dots,G_\pt)$. Furthermore, denote by $Y_{t;i}^{\prime(k)}=G_i(\eps_t,\eps_{t-1},\dots,\eps_{t-k+1},\eps_{t-k}^\prime,\eps_{t-k-1},\eps_{t-k-2},\dots)$ the process where $\eps_{t-k}$ is replaced by an i.i.d. copy of it.
Furthermore, define the physical dependence coefficients, see \cite{wu2005nonlinear,zhang2017gaussian,liu2013probability}, in the following way. Let $\delta_{k,q,i}=\|Y_{0;i}-Y_{0;i}^{\prime(k)}\|_q, k \geq 0$, $\Delta_{m,q;i}=\sum_{k=m}^\infty \delta_{k,q;i}$,$\|Y_{;i}\|_{q,\alpha}=\sup_{m\geq 0} (m+1)^\alpha \Delta_{m,q;i}$ $\Psi_{q,\alpha}=\max_{1\leq i\leq \pt} 
\|Y_{\cdot;i}\|_{q,\alpha},$
$\Upsilon_{q,\alpha}=(\sum_{i=1}^\pt (\sup_{m\geq 0} (m+1)^\alpha \Delta_{m,q;i})^q)^{1/q}$, 
$\omega_{k,q}=\| \max_{1\leq i \leq \pt} \| Y_{t;i}-Y_{t;i}^{\prime(k)}\|_q$, $\Omega_{m,q}=\sum_{k=m}^\infty \omega_{k,q},$ and 
$\| \|Y_\cdot\|_\infty \|_{q,\alpha}=\sup_{m\geq 0} (m+1)^\alpha \Omega_{k,q}$.
Furthermore, let $\nu_q=\sum_{j=1}^\infty (j^{q/2-1} \omega_{k,q})^{1/(q+1)}$.

If Assumption~\ref{ass1}\eqref{ass1.1},\eqref{ass1.2},\eqref{ass1.6} is satisfied, then the following assertions hold true:
\begin{enumerate}
    \item For the process $\{(e_r^\top (\Gammas)^{-1} W_t W_t^\top e_r-1)_{r=1,\dots,dp},t \in \Z\}$ we have $\nu_{q/2}<\infty$.
    \item Let ${\Gn}=\{(j_1,r_1,s_1),\dots,(j_\G,r_\G,s_\G)\}$ be some set of indices and let further    $Y_{t,i}:=e_{(s_i-1)p+r_i} (\Gammas)^{-1} W_t \eps_{t+1;j_i}/s.e.(r_i,j_i,s_i), i =1,\dots,\G$. Then for some $\alpha>1/2-1/q$ we have $\Psi_{q,\alpha}<\infty, \Upsilon_{q,\alpha}=O(\G^{1/q}), \| \|Y_\cdot\|_\infty \|_{q,\alpha}=O(\kpp \kp), \linebreak\nu_q<\infty$.
\end{enumerate}
\end{lem}

\begin{lem} \label{lem.concen}
Let $W_{t}=(X_{t}^\top,\dots,X_{t-d+1}^\top)^\top=\sum_{j=0}^\infty \A^j \U_{t-j}$ and suppose that   Assumption~\ref{ass1}\eqref{ass1.1},\eqref{ass1.2},\eqref{ass1.5}, \eqref{ass1.6} is true. Then, for some vector $v \in \R^{dp}, \|v\|_1=1$ we have 
$$
\max_{r=1,\dots,dp}| \frac{1}{n} e_r^\top \Big[\sum_{t=d}^{n-1} (\Gammas)^{-1} W_t W_t^\top -I_{dp}\Big] v| =O_P(\sqrt{\log(p) /n}).$$
\end{lem}

\begin{lem}\label{lem.max}
Let $\{Y_t, t \in \Z\}$ be a centered,  $p$-dimensional stationary process and let $\nu<\infty$ as given in Lemma~\ref{lem.physical}. If 
$p /((\log p)^{q/2}n^{q/2-1})=O(1)$, then\\
$
\max_{1 \leq j \leq p} | e_j^\top 1/\sqrt n \sum_{t=1}^n Y_t|=O_P( \sqrt{\log p}).
$
\end{lem}
\begin{proof}
The assertion follows  by  Nagaev's inequality for physical dependent processes, see Theorem 2 in \cite{liu2013probability}, and by using the same arguments as those  used in Lemma~\ref{lem.concen}.
\end{proof}

\begin{lem} \label{lem.est.gamma}
Let $\mathds{E}=(I_p, 0,\dots,0)^\top \in \R^{dp\times p}$
If Assumption~\ref{ass1}\eqref{ass1.1} to \eqref{ass1.4} holds true, then the  estimator of $\Gamma(h)$ given by 
\begin{align}
\label{eq.Ghat-thr}
    \hat \Gamma^{(st)}(h)&=\sum_{j=0}^\infty \hat {\mathds{A}}^{j+h}\mathds{E} \Sigmah \mathds{E}^\top (\hat{\mathds{A}}^\top)^j=\hat {\A}^{h}\operatorname{vec}^{-1} \Big((I_{(dp)^2}-\hat {\mathds{A}}\otimes\hat {\mathds{A}})^{-1} \operatorname{vec}(\mathds{E}\hat \Sigma_{\varepsilon})\Big), h\geq 0,
\end{align}
and $\hat \Gamma^{(st)}(h)=\big(\hat \Gamma^{(st)}(-h)\big)^\top$ for $h<0$ satisfies $\|\hat \Gamma^{(st)}(h)-\Gamma^{(st)}(h)\|_\infty=O_P(\cp \kp^4 \linebreak \sqrt{\gp/n}).$ Furthermore, if Assumption~\ref{ass1}\eqref{ass1.5} holds true, we also have 
$\|(\hat \Gamma^{(st)}(h))^{-1}-( \Gamma^{(st)}(h))^{-1}\|_\infty=O_P(\cp \kp^4 \kpp^2\sqrt{\gp/n}).$
\end{lem}

\begin{proof}[Proof of Lemma~\ref{lem.est.gamma}]
In order to simplify  notation we set $h=0$. Note that for $j \in \N$ we have 
$
\hat \A^j -\A^j=(\hat A-A+A)(\hat A^{j-1}-A^{j-1})+(A-\hat A)A^{j-1}$. Using this recursive formula, we obtain
$
\hat \A^j -\A^j=\sum_{s=0}^{j-1}[(\hat \A-\A)+\A]^s(\hat \A-\A)\A^{j-1-s}.
$
Since $\|\A-\hat \A\|_\infty=O_P(\kp \sqrt{\gp/n})$ by Assumption~\ref{ass1}\eqref{ass1.3}, we have
$\|\hat \A^j -\A^j\|_\infty=\|\sum_{s=0}^{j-1}[(\hat \A-\A)+\A]^s(\hat \A- \A)\A^{j-1-s}\|_\infty=O_P(j \lambda^{j-1} \kp^3 \sqrt{\gp/n})$. 

Since $\|{\A}\|_\infty=\|{\A}^\top\|_1$ we have by the same arguments $\|({\A}^\top)^j-(\hat {\A}^\top)^j\|_\infty=  O_P(j \lambda^{j-1} \linebreak  \kp^3 \sqrt{\gp/n})$. 
Furthermore, note that Assumption~\ref{ass1}\eqref{ass1.2} implies $\sum_{j=0}^\infty \|{\A}^j\|_g=O(\kp)$ and also $\sum_{j=0}^\infty \|\hat \A^j-{\A}^j\|_g=O(\kp^3 \sqrt{\gp/n})$ for $ g \in \{1,\infty\}$   We then have $\|\Gammas-\Gammah\|_\infty \leq \|\sum_{j=0}^\infty ({\A}^j-\hat {\A}^j) \mathds{E}\Sigmaeps\mathds{E}^\top ({\A}^\top)^j\|_\infty+\|\sum_{j=0}^\infty ({\hat {\A}}^j) \mathds{E}(\Sigmaeps-\hat \Sigma_\eps)\mathds{E}^\top ({\A}^\top)^j\|_\infty+\|\sum_{j=0}^\infty ({\hat {\A}}^j)\mathds{E} \hat\Sigma_\eps\mathds{E}^\top \! (({\A}^\top)^j-(\hat {\A}^\top)^j)\|_\infty=O_P(\kp^4 \cp \sqrt{\gp/n})$.

Using  $A^{-1}-B^{-1}=A^{-1}(B-A)B^{-1}$, we have
$\|(\Gammas)^{-1}-(\Gammah)^{-1}\|_\infty\leq\|(\Gammas)-{\Gammah}\|_\infty (\|(\Gammas)^{-1}-(\Gammah)^{-1}\|_\infty +\|(\Gammas)^{-1}\|_\infty) \|\Gamma^{-1}\|_\infty$. For $n$ large enough such that $\|(\Gammas)-{\Gammah}\|_\infty<\|(\Gammas)^{-1}\|_\infty^{-1}$, we have\\
$\|(\Gammas)^{-1}-(\Gammah)^{-1}\|_\infty\leq \|(\Gammas)-{\Gammah}\|_\infty {\|(\Gammas)^{-1}\|_\infty^2}\linebreak ({1-\|(\Gammas)^{-1}\|_\infty\|(\Gammas)-{\Gammah}\|_\infty})^{-1}$
\end{proof}

\begin{lem} \label{lem.Aint}
Let $\hat A^{(s)},s=1,\dots,d$ be the estimator of $A^{(s)},s=1,\dots,d$ given 
by
 \begin{equation*} 
   \hat A_{j;r}^{(s)}  = \hat A_{j;r}^{(s,1)} \ind{\big(| \hat A_{j;r}^{(s,1)} | \geq a_n,|\hat A_{j;r}^{(s,2)} | \geq a_n \big)}, j,r=1,\dots,p,s=1,\dots,d.
 \end{equation*}
If 
$\max_{i,j,s} |e_i^\top (\hat A^{(s,1)}-A)e_j|=O_P(r_n)$, $\sum_{s=1}^d \|\hat A^{(s,1)}-A\|_\infty=O_P(\kp r_n)$,
$\sum_{s=1}^d \|\hat A^{(s,2)}-A\|_1=O_P(\kp r_n)$ and $a_n=C r_n$ for some constant $C>0$, then
$\sum_{s=1}^d \|\hat A^{(s)}-A\|_\infty=O_P(\kp r_n)$,
$\sum_{s=1}^d \|\hat A^{(s)}-A\|_1=O_P(\kp r_n)$,
$\max_j \sum_{s=1}^d \|\hat A^{(s)}e_j\|_0=O_P(\kp)$ and 
$\max_j \sum_{s=1}^d \|e_j^\top \hat A^{(s)}\|_0=O_P(\kp)$.
\end{lem}

\begin{proof}
The proof uses ideas similar to those used   in the proof of Theorem 1 in \cite{bickel2008}. Let $A^{(s,thr)}=A^{(s)} =\big(A^{(s)}_{i;j}\ind (|A_{i;j}^{(s)}|\geq a_n)\big)_{i,j=1,\ldots,p}$, $ s=1,\dots,d$, be the thresholded version of $A^{(s)}$. Since $\kp=\max_{1\leq j \leq p} \big\{\sum_{s=1}^d\|e_j^\top A^{(s)}\|_0,\sum_{s=1}^d\| A^{(s)}e_j\|_0\big\}$, we have $\|A^{(s,thr)}-A^{(s)}\|_1\leq \kp a_n=O(\kp r_n)$ and $\|A^{(s,thr)}-A^{(s)}\|_\infty=O(\kp r_n)$, which 
gives the resulting rates for $ \|A^{(s,thr)}-A^{(s)}\|_1$ and for $\|A^{(s,thr)}-A^{(s)}\|_\infty $, when $ A^{(s)}$ is thresholded using the thresholding parameter $a_n$ which satisfies $a_n=O(r_n)$.
Let $J_{i,s}=\{j=1,\dots,p : A_{i,j}^{(s)})\not=0\}$ be the set of indices for which the entries of the $i$th row of $A^{(s)}$ are non-zero. We then have  
$
\sum_{s=1}^d \max_i \sum_{j=1}^p \ind\{|\hat A^{(s)}_{i;j}| \geq a_n\}  \leq \sum_{s=1}^d \max_i \sum_{j \in J_{i,s}} \ind\{|\hat A^{(s,1)}_{i;j}| \geq a_n\}+\sum_{j \not \in J_{i,s}} \ind\{|\hat A^{(s,1)}_{i;j}|\geq a_n\}
\leq \kp + \sum_{s=1}^d  \linebreak \max_i \sum_{j \not \in J_{i,s}} |\hat A^{(s,1)}_{i,j}-A_{i;j}^{(s)}|/a_n\leq \kp + \sum_{s=1}^d \|\hat A^{(s,1)}-A{(s)}\|_\infty/a_n=O_P(\kp).
$
Similarly by  the properties of $\hat A^{(s,2)}$, we have $\sum_{s=1}^d \max_j \sum_{i=1}^p \ind\{|\hat A^{(s)}_{i;j}|\geq a_n\}=O_P(\kp).$ Thus, $\hat A^{(s)}$ gives a row-wise and column-wise sparse estimator of $ A^{(s)}$. 
Moreover, 
$   \sum_{s=1}^d \linebreak \| \hat A^{(s)}-A^{(s,thr)}\|_\infty \leq \sum_{s=1}^d \max_i \sum_{j=1}^p |\hat A^{(s,1)}_{i;j}| \ind\{|\hat A^{(s,1)}_{i;j}|\geq a_n,|\hat A^{(s,2)}_{i;j}|\geq a_n,  |A^{(s)}_{i;j}|< a_n\} 
    +\sum_{s=1}^d \max_i \sum_{j=1}^p | A^{(s)}_{i;j}| \ind\{(|\hat A^{(s,1)}_{i;j}|< a_n \text{ or } |\hat A^{(s,2)}_{i;j}|< a_n),  |A^{(s)}_{i;j}|\geq a_n\}
    +\sum_{s=1}^d \max_i \linebreak \sum_{j=1}^p |\hat A^{(s,1)}_{i;j}-A_{i;j}| \ind\{|\hat A^{(s,1)}_{i;j}|\geq a_n,|\hat A^{(s,2)}_{i;j}|\geq a_n,  |A^{(s)}_{i;j}|\geq a_n\}
    =I+II+III.$
The indicator functions  ensure that each part consist only of $O_P(\kp)$ non-zero terms. Due to $\max_{i,j}\sum_{s=1}^d |e_i^\top (\hat A^{(s,1)}-A^{(s)})e_j|=O_P(r_n)$ each non-zero terms is of order $O_P(r_n)$ and  we obtain $\sum_{s=1}^d\|\hat A^{(s)}-A^{(s,thr)}\|_\infty=O_P(\kp r_n)$. In the same way, we obtain $\sum_{s=1}^d\| \hat A^{(s)}-A^{(s,thr)}\|_1=O_P(\kp r_n)$ and the assertion follows.
\end{proof}

\begin{lem} \label{lem.de.approx}
Let ${\Gn}=\{(j_1,r_1,s_1),\dots,(j_\G,r_\G,s_\G)\}$ be some set of indices and let   $Y_{t;(r_i,j_i,s_i)}=e_{(s_i-1)p+r_i} (\Gammas)^{-1} W_t \eps_{t+1;j_i}/s.e.(r_i,j_i,s_i), i =1,\dots,\G$. Set $h_1(x)=x, h_2(x)=-x$ and $ h_3(x)=|x|$. 
If Assumption~\ref{ass1}\eqref{ass1.1}-\eqref{ass1.6} hold true, then 
$$\max_{(j,r,s) \in {\Gn}} \frac{\sqrt{n}h_i(\hat A_{j;r}^{(s,de)}-A_{j;r}^{(s)})}{\widehat {s.e.} (j,r,s)}=\max_{(j,r,s) \in {\Gn}} h_i(1/\sqrt{n}\sum_{t=d}^{n-1}   Y_{t;(r,j,s)})+o_P(1),i=1,2,3.$$
\end{lem}
\begin{proof}
We proof the assertion for $h_3$, the other cases can be handled analogously. Denote by $\mathcal{I}=\{1,\dots,\G\}$ a set of indices corresponding to ${\Gn}$ and let $\tilde r=(s-1)p+r$ for $(j,r,s)\in {\Gn}$. We have  $(\hat A_{j;r}^{(s,de)}-A_{j;r}^{(s)}){\widehat {s.e.} (j,r,s)}=\sqrt{n}(\sum_{t=d}^{n-1} \hat Z_{t;\tilde r} W_{t;\tilde r})^{-1}
(\sum_{t=d}^{n-1} \hat Z_{t;\tilde r} \eps_{t+1;j}) \linebreak
+\sqrt{n}(\sum_{t=d}^{n-1} \hat Z_{t;\tilde r} W_{t;\tilde r})^{-1}
(\sum_{t=d}^{n-1} \hat Z_{t;\tilde r} W_{t;-(\tilde r)}(\Xi_{j;-(\tilde r)}-\hat \Xi_{j;-(\tilde r)}^{(re)})$

Consider first that $\hat Z_{t;\tilde r}=(e_{\tilde r} (\Gammah)^{-1} e_{\tilde r})^{-1} e_{\tilde r}^\top (\Gammah)^{-1} W_t=:\omega_{\tilde r} e_{\tilde r}^\top (\Gammah)^{-1} W_t$. The positive factor $\omega_{\tilde r}$ occurs in both terms in the numerator and denominator and can be omitted.  Let $DN=1/n(\sum_{t=d}^{n-1} e_{\tilde r}^\top (\Gammah)^{-1} W_t W_{t;\tilde r})$. Furthermore, observe that
$
DN=1 /n \sum_{t=d}^{n-1} e_{\tilde r}^\top (\Gammah)^{-1} W_t W_t ^\top e_{\tilde r}
=
1/n \Big( \sum_{t=d}^{n-1} e_{\tilde r}^\top (\Gammas)^{-1} W_t W_t ^\top e_{\tilde r} +  \sum_{t=d}^{n-1} e_{\tilde r}^\top ((\Gammah)^{-1}- (\Gammas)^{-1}) W_t W_t ^\top e_{\tilde r}\Big)=1+O_P(\kppp \sqrt{\gp/n}),
$
due to Lemma~\ref{lem.concen} and Lemma~\ref{lem.est.gamma}. Note further that $s.e.(j,r,s)^2
={\Sigma_{\eps,j;j}}(e_{(s-1)p+r}^\top \linebreak  (\Gammas)^{-1} e_{(s-1)p+r})$ and 
$
\widehat s.e.(j,r,s)^2
={\hat \Sigma_{\eps,j;j}}(e_{(s-1)p+r}^\top (\Gammah)^{-1} e_{(s-1)p+r})$. Thus, $\max_{(j,r,s) \in {\Gn}} \widehat s.e.(j,r,s)^2/s.e.(j,r,s)^2=1+O_P(\cp \sqrt{\gp/n}+\kppp  \linebreak \sqrt{\gp/n})$ by Assumption~\ref{ass1}\eqref{ass1.4} and Lemma~\ref{lem.est.gamma}.

Furthermore, we have
\begin{align*}
&\max_{(j,r,s)  \in G}  |\frac{1}{\sqrt{n}}(\sum_{t=d}^{n-1} e_{\tilde r} (\Gammah)^{-1} W_t  W_{t;-(\tilde r)}(\Xi_{j;-(\tilde r)}-\hat \Xi_{j;-(\tilde r)}^{(re)})| \\
&\leq 
\max_{(j,r,s) \in G}| \frac{1}{\sqrt{n}}(\sum_{t=d}^{n-1} e_{\tilde r} (\Gammas)^{-1} W_t  W_{t;-(\tilde r)}(\Xi_{j;-(\tilde r)}-\hat \Xi_{j;-(\tilde r)}^{(re)})| \\
&\qquad+
\max_{(j,r,s) \in G}| \frac{1}{\sqrt{n}}(\sum_{t=d}^{n-1} e_{\tilde r} ((\Gammah)^{-1}- (\Gammas)^{-1}) W_t  W_{t;-(\tilde r)}(\Xi_{j;-(\tilde r)}-\hat \Xi_{j;-(\tilde r)}^{(re)})| \\
&\leq \max_{(j,r,s) \in G}  \max_{s \not= \tilde r} | \frac{1}{\sqrt{n}}\sum_{t=d}^{n-1} e_{\tilde r} (\Gammas)^{-1} W_t  W_{t} e_s|
\| \Xi -\hat \Xi^{(re)}\|_\infty \\
& \qquad + \sqrt{n}\max_{s} | 1/n \sum_{t=d}^{n-1} W_{t;s} W_{t;s}| \| (\Gammah)^{-1}- (\Gammas)^{-1}\|_\infty \|\Xi -\hat \Xi^{(re)}\|_\infty \\
&=O_P( \kp \sqrt{\log(p) \gp}/n+ \kp \kppp \gp/\sqrt{n}),
\end{align*}
by Lemma~\ref{lem.concen},\ref{lem.est.gamma} and \ref{lem.max}.
Furthermore, we have 
$
1/\sqrt{n} \sum_{t=d}^{n-1}e_{\tilde r} (\Gammah)^{-1} W_t \eps_{t+1;j}= \linebreak
1/\sqrt{n} \sum_{t=d}^{n-1}e_{\tilde r} (\Gammas)^{-1} W_t \eps_{t+1;j}+
1/\sqrt{n} \sum_{t=d}^{n-1}e_{\tilde r} ((\Gammah)^{-1}-(\Gammas)^{-1}) W_t \eps_{t+1;j}.
$
 Applying Lemma~\ref{lem.physical} to  $\{W_t \eps_{t+1}\}$ leads to $\nu_q<\infty$. 
Lemma~\ref{lem.max} implies \linebreak
$ \max_{(j,r,s) \in {\Gn}} |1/\sqrt{n} \sum_{t=d}^{n-1}e_{\tilde r} ((\Gammah)^{-1}-(\Gammas)^{-1}) W_t \eps_{t+1;j}|
 \leq 
\|(\Gammah)^{-1}- \linebreak(\Gammas)^{-1}\|_\infty \max_{(j,r,s) \in {\Gn}} \max_{s_1,s_2} | 1/\sqrt{n} \sum_{t=d}^{n-1} e_{s_1} W_t \eps_{t+1;s_2} |
= O_P(\sqrt{\log(\G)} \linebreak \kppp \sqrt{\gp/n}).
$ This completes the calculations to show that the nuisance terms are asymptotically negligible and the assertion follows.
\end{proof}

{\it Proof of Theorem~\ref{thm.clt}:}
Applying Lemma~\ref{lem.de.approx} gives us 
$
\sqrt{n} (\hat A_{j;r}^{(s,de)}-A_{j;r}^{(s)}) =1/\sqrt{n} \linebreak \sum_{t=d}^{n-1}e_{(s-1)p+r}^\top (\Gammas)^{-1} W_t \eps_{t+1;j}+o_P(1)
=:1/\sqrt{n} \sum_{t=d}^{n-1} Y_t +o_P(1).$ Furthermore, we have by Lemma~\ref{lem.physical} that $\{Y_t\}$ fulfills the conditions of physical dependence, hence $\nu_q<\infty$. Note further that due to the i.i.d property of $\{\eps_t\}$ we have that $\{Y_t\}$ is an uncorrelated sequence which gives $\var 1/\sqrt{n} \sum_{t=d}^{n-1} Y_t=n/(n-d) \Sigma_{\eps;jj} (e_{(s-1)p+r}^\top (\Gammas)^{-1}e_{(s-1)p+r}^\top)$. Furthermore, Assumption~\ref{ass1}\eqref{ass1.6} ensures the existence of fourth order moments. Thus a Lyapounov condition can be verified for the sequence $\{n^{-1/2}\sum_{t=d+1}^n Y_t \}$,  establishing $1/\sqrt{n} \sum_{t=d+1}^n Y_t \overset{D}{\to} \mathcal{N}(0,\Sigma_{\eps;jj} (e_{(s-1)p+r}^\top (\Gammas)^{-1}e_{(s-1)p+r}^\top))$   via an extension  of the central limit theorem for functional dependent random variables,  Theorem 3 of  \cite{wu2011asymptotic},   to triangular arrays;  see also  Theorem 27.3 of 
\cite{billingsley1995probability}. Thus the assertion follows by Slutsky's Lemma.
\hfill $\Box$

\noindent{\it Proof of Lemma~ \ref{lem.dependence}:}
Let $\tilde r_1=(s_1-1)p+r_1$ and $\tilde r_2=(s_2-1)p+r_2$. Note further that due to the i.i.d property of $\{\eps_t\}$ we have that $\{W_{t_1} \eps_t\}$ is an uncorrelated sequence. Thus, we have
$
\cov(e_{\tilde r_1}^\top (\Gammas)^{-1} W_{t-1} \eps_{t;j_1},e_{\tilde r_2} (\Gammas)^{-1} {t-1} \eps_{t;j_2})
=(\Sigma_{\eps;jj} e_{\tilde r_1}^\top (\Gammas)^{-1} e_{\tilde r_2}$, the assertion follows  Lemma~\ref{lem.de.approx}. \hfill $\Box$

\noindent{\it Proof of Theorem~\ref{thm.bootstrap.lasso}:}
Denote by $\,^*$ the corresponding quantities in the bootstrap world, that is  $X_t^*, W_t^*=(X_{t-1}^*,\dots,X_{t-d}^*)$ and $\eps_t^*, t=1,\dots,n$.
We show  that,  as $n\to \infty$, 
 $\sqrt n (\hat A_{j;r}^{*(s,de)} - \hat A_{j;r}^{(s)}) \overset{D}{\to} \mathcal{N}\big(0,s.e.(j,r,s)^2\big)\text{ in probability},$
from which  the assertion  follows by the triangular inequality.
To establish the above weak convergence, 
 similar  ideas as  those used in the proof of Theorem \ref{thm.clt} can be applied. 
 
 Since $\hat A^{(s)},s=1,\dots,d$ and $\hat \Sigma_\eps$ give an sparse estimate, see Lemma~\ref{lem.Aint} and Lemma~\ref{lem.est.var.eps} we have that $\{X_t^*\}$ is a sparse VAR$(d)$ process with Gaussian i.i.d. innovations. Furthermore, due to Assumption~\ref{ass1}\eqref{ass1.2} and \eqref{ass1.3} this process is (asymptotically) causal and we have $X_t^*=(\hat A^{*1},\dots,\hat A^{*d}) (X_{t_1}^{*\top},\dots,X_{t-d}^{*\top})^\top+\eps_t^*=:\Xi^* W_t^* +\eps_t^*$. Thus, Lemma~\ref{lem.concen},\ref{lem.max}, and \ref{lem.physical} hold for this processes as well. A regularized estimator $\hat \Xi^{*(re)}=(\hat A^{*(s,re)},s=1,\dots,d)$ similarly as in the construction of $\{\hat A^{(de)}\}$ but based on the pseudo time series
 $X_1^*,\dots,X_n^*$ could be used to estimate $\Xi^*=(\hat A^{(s)},s=1,\dots,d)$ in the bootstrap world. By Lemma~\ref{lem.Aint} we obtain for $l \in \{1,\infty\}$ $\|\Xi^{*(re)}-\Xi^{*}\|_l=O_{P^*}(\kp \sqrt{\log(p)/n})$, where $\log(p)$ appears due to the Gaussianity of $\{\eps_t^*\}$. Assumption~\ref{ass1}\eqref{ass1.3} gives $\sum_{s=1}^d \|\hat A^{(s)}-A^{(s)}\|_l=\|\Xi^*-\Xi\|_l=O_P(\kp \sqrt{\gp/n})$. Hence, by the triangular inequality we obtain $\|\Xi^{*(re)}-\Xi\|_l=O_{P}(\kp \sqrt{\gp/n})$. Similarly, we obtain $\|\hat \Sigma_{\eps}^{*}-\Sigma_{\eps}\|_\infty=O_{P}(\cp \sqrt{\gp/n})$ by Assumption~\ref{ass1}\eqref{ass1.4} and Lemma~\ref{lem.est.var.eps}. This implies that the estimator $\Gammahast$ which is similarly constructed as $\Gammah$ in \eqref{eq.GammaHat} but based on $(\hat A^{*(s,re)},s=1,\dots,d)$ and $\hat \Sigma_{\eps}^{*}$ fulfills by Lemma~\ref{lem.est.gamma}
the rate $\|(\Gammahast)^{-1}-(\Gammas)^{-1}\|_\infty=O_{P}(\kppp \sqrt{\gp/n})$. This implies
$
\Sigma_{\eps;jj}^* (e_{(s-1)p+r}^\top \Gammasast(e_{(s-1)p+r}=\Sigma_{\eps;jj} (e_{(s-1)p+r}^\top \Gammas(e_{(s-1)p+r}+o_P(1).
$
 Since the Gaussian i.i.d. pseudo innovations $\eps^*_t$ fulfill the moment condition of Assumption~\ref{ass1}\eqref{ass1.6} and the bootstrap analogue of the estimators fulfill the required rates, we obtain by the same arguments as those used in the proof of Theorem~\ref{thm.clt} that, as $n\to \infty$, 
 $\sqrt n (\hat A_{j;r}^{*(s,de)} - \hat A_{j;r}^{(s)}) \overset{D}{\to} \mathcal{N}\big(0,s.e.(j,r,s)^2\big)$ in probability.
\hfill $\Box$

\noindent{\it Proof of Corollary~\ref{co.stud}:}
Follows by Lemma~\ref{lem.dependence} and the arguments used in the proof of Theorem~\ref{thm.bootstrap.lasso} which gives 
$\|\Sigma_{\eps}^*-\Sigma_{\eps}\|_\infty=o_P(1)$ and $\|\Gammahast-\Gammas\|_\infty=o_P(1)$.
\hfill $\Box$

\begin{proof}[Proof of Theorem~\ref{thm.bootstrap.test}]
By  the Cram\'er-Wold device and the 
same arguments as those used in the proof of Theorem~\ref{thm.clt}, it follows that  
the vector 
$\big(\sqrt{n}|\hat A_{j;r}^{(s,de)}|/\widehat {s.e.} (j,r,s): (j,r,s)\in G \big)^\top$
converges weakly  to a $|G|$-dimensional normal distribution with zero mean vector and covariance matrix having components those given  in  Lemma \ref{lem.dependence}. Then, following the same arguments as those used in the proof  of Theorem \ref{thm.bootstrap.lasso} and  by  Lemma \ref{lem.dependence2}, 
the same result can be established for the  limiting distribution of  the 
bootstrap analogue 
$\big(\sqrt{n}|\hat A_{j;r}^{*(s,de)}|/\widehat {s.e.}^* (j,r,s): (j,r,s)\in G \big)^\top$.
The assertion of the theorem  follows then by the continuous mapping theorem.
\end{proof}

\begin{lem}\label{lem.est.var.eps}
Let 
$\hat \eps_t= \tilde \eps_t- 1/{n-d}\sum_{t=d}^n \tilde \eps_t, \tilde \eps_t=X_t-\sum_{s=1}^d \hat A^{(s)} X_{t-j}=X_t-\Xi W_{t-1}$, $t=d+1,\dots,n$, be the centered estimated residuals of $\eps_1 \sim (0,\Sigma)$,
 $\max_j \|\Sigma e_j\|_0=\cp$ and $ \|\Sigma\|_1=O(\cp)$. Recall the definitions of $\widetilde{\Sigma}_\eps$ and $ \hat{\Sigma}_\eps$ given in (\ref{eq.sigma}). 
If Assumption 1\eqref{ass1.1} and 1\eqref{ass1.3} are satisfied, then 
\begin{align} 
\max_{i,j} |e_i^\top (\widetilde \Sigma_{\eps}- \Sigma_{\eps,n}) e_j|=O_P\Big(\sqrt{\gp/n}\kp^2(\sqrt{\gp/n}+\sqrt{\log(p)/n})\Big), \label{estimated.residuals}
\end{align}
where $ \Sigma_{\eps,n}=n^{-1}\sum_{t=1}^n \eps_t \eps_t^\top$.
If the $\eps_t$'s are   Gaussian and $b_n=C(\sqrt{\log(p)/n})$ for some constant $C>0$, then
\begin{align}\label{eq.est.sigma}
\|\hat \Sigma_{\eps}-\Sigma_\eps\|_1=O_P\Big(\cp \sqrt{\log(p)/n}+\kp^2 \log(p)/n\Big)
\end{align}
and
$
    \max_{1\leq j \leq p} \| \hat \Sigma_{\eps} e_j \|_0 = O_P(\cp).
$
\end{lem}
\begin{proof}
In order to simplify notation, let $1/(n-d)\sum_{t=d}^n \tilde \eps_t=0$. and  we assume that we have observations $X_{-d+1},\dots,X_n$. Note that
$    \max_{i,j} \|e_i^\top (\tilde \Sigma_{\eps}-\Sigma_{\eps,n}) e_j\| \leq \max_{i,j}| 1/n \linebreak \sum_{t=1}^n\eps_{t;i} W_{t-1}^\top (\Xi_i-\hat \Xi_j)| + \max_{i,j}| 1/n \sum_{t=1}^n\eps_{t;j} W_{t-1}^\top (\Xi_i-\hat \Xi_i)| +
    \max_{i,j} | (\Xi_i-\hat \Xi_i)^\top \linebreak(1/n \sum_{t=1}^n W_{t-1} W_{t-1}^\top) (\Xi_j-\hat \Xi_j) = I+II+III.
$
By Lemma~\ref{lem.max} and the convergence rate  of $\hat A$, we have   that $I\leq \max_i \max_k |1/n \sum_{t=1}^n \eps_{t;i} W_{t_1;k}| \|\Xi-\hat \Xi\|_\infty=O_P(\kp/n\sqrt{\log(p)} \linebreak \sqrt{\gp})$. The same arguments can be applied to $II$. Furthermore, we have by the convergence rate of   $\hat A$,  that $III=O_P(\kp^2 \gp/n)$. This establishes  (\ref{estimated.residuals}).
In the Gaussian case we have $\max_{i,j} | e_i^\top(\widetilde \Sigma_{\eps}-\Sigma_{\eps,n}) e_j|=O_P(\log(p)/n\kp^2)$,  and  from Theorem 1 in \cite{bickel2008} we get that  $\max_{i,j} | e_i^\top (\Sigma_{\eps,n}-\Sigmaeps) e_j|=O_P(\sqrt{\log (p)/n})$. Note further that $\max_i \sum_{j=1}^p \ind\{e_i^\top\Sigmaeps e_j\}=\cp$. To establish (\ref{eq.est.sigma}), we mainly follow the arguments given in the proof of Theorem 1 in \cite{bickel2008}. Since $b_n=C(\sqrt{\log(p)/n})$, we have $\|\Sigmaeps-\Sigma_{\eps}^{(thr)}\|_1=O_P(\cp \sqrt{\log(p)/n} )$, where $\Sigma_{\eps}^{(thr)}$ denotes the  true covariance matrix $\Sigma_\eps$  thresholded with threshold parameter  $b_n$.
Let $\hat \sigma_{ij}=e_i^\top \tilde \Sigma_\eps e_j$ and $\sigma_{i,j}=e_i^\top \Sigmaeps e_j$. Then, 
$    \|\Sigmahthr-\Sigma_{\eps}^{(thr)}\|_1 \leq \max_i \sum_{j=1}^p |\hat \sigma_{ij}| \ind\{|\hat \sigma_{ij}|\geq b_n, | \sigma_{ij}|<b_n \}+\max_i \sum_{j=1}^p | \sigma_{ij}| \ind\{|\hat \sigma_{ij}|< b_n, | \sigma_{ij}| \geq b_n \}
    +\max_i \sum_{j=1}^p |\hat \sigma_{ij} - \sigma_{ij}| \ind\{|\hat \sigma_{ij}|\geq b_n, | \sigma_{ij}| \geq b_n \}= IV+V+VI.
$
Equation (\ref{estimated.residuals}) and $\max_i \sum_{j=1}^p \ind\{e_i^\top\Sigmaeps e_j\}=\cp$ ensures that 
$V=O_P(\cp \sqrt{\log(p)/n})$ and $VI=O_P(\cp \sqrt{\log(p)/n})$. Furthermore, we have for some $c \in (0,1)$, that 
$IV=\max_i \sum_{j=1}^p |\hat \sigma_{ij}-\sigma_{ij}| \ind\{|\hat \sigma_{ij}|\geq b_n, | \sigma_{ij}|= 0 \}+O_P(\cp\sqrt{\log(p)/n})=\max_i \sum_{j=1}^p |\hat \sigma_{ij}-\sigma_{ij}| \ind\{|\hat \sigma_{ij}-E_{ij}|\geq c b_n, | \sigma_{ij}|= 0 \}+\max_i \sum_{j=1}^p |\hat \sigma_{ij}-\sigma_{ij}| \ind\{|E_{ij}|\geq (1-c) b_n, | \sigma_{ij}|= 0 \}+O_P(\sqrt{\log(p)/n}\kp)=VII+VIII+O_P(\cp \sqrt{\log(p)/n}),$ where 
$E_{ij}=e_i^\top(\widetilde \Sigma_{\eps}-\Sigma_{\eps,n}) e_j$. Note that $\hat \sigma_{ij}-E_{ij}=e_i^\top \Sigma_{\eps,n} e_j$ and thus, $VII$ can be bounded analogously as part $IV$ in the proof of Theorem 1 in \cite{bickel2008}. Since $\max_{i,j}|E_{i,j}|=O_P(\cp \log(p)/n )=O_P(b_n^2 \cp)$, we have
$P(\max_{i,j} | E_{i,j}| \geq (1-c) b_n)=o(1)$. Consequently, $VII=o(\cp \sqrt{\log(p)/n})$ and equation (\ref{eq.est.sigma}) follows. Let $J_i=\{j=1,\dots,p : \Sigma_{\eps,i,j}\not=0\}$ be the set of indices for which the entries of the $i$th row of $\Sigmaeps$ are non-zero. Then, for $0<b_n=C(\sqrt{\log(p)/n})$, we have that    $\max_i \|{\hat \Sigma_{\hat \eps}} e_i\|_0= \max_i \sum_{j \in J_i} \ind(|{\hat \Sigma_{\hat \eps,i,j}}{}|>b_n)+\sum_{j \not \in J_i} \ind(|{\hat \Sigma_{\hat \eps,i,j}}|>b_n)\leq \cp + \max_i \sum_{j \not \in J_i} \ind(|{\hat \Sigma_{\hat \eps,i,j}}-\Sigma_{\eps,i,j}|>b_n)\leq \cp + b_n^{-1} \|{\hat \Sigma_{\hat \eps}}-\Sigmaeps\|_1=O_P(\cp)$.
\end{proof}

\begin{proof}[Proof of Theorem \ref{thm.test.grow}]
Let $\G=|\Gn|$ and $\mathcal{I}=\{1,\dots,\G\}$ be a set of indices corresponding to $\Gn$. 
We show this by using first the triangular inequality, thus
\begin{align*}
\sup_{c \in \R}& | P^\ast( \max_{(j,r,s) \in \Gn} | \frac{\sqrt{n}|\hat A_{j;r}^{*(s,de)}|}{\widehat {s.e.}^* (j,r,s)}| \leq c )-
P( \max_{(j,r,s) \in \Gn} | \frac{\sqrt{n}|\hat A_{j;r}^{(s,de)}|}{\widehat {s.e.} (j,r,s)}| \leq c )| \\
\leq&
\sup_{c \in \R} | P(\max_{(j,r,s) \in {\Gn}} \sqrt{n}|\hat A_{j;r}^{(s,de)}|/\widehat {s.e.} (j,r,s)) - P( |\max_{i \in \mathcal{I}}  e_i^\top D_0^{-1} Z e_i| \geq c)| \\
&+
\sup_{c \in \R} | P^\ast( \max_{(j,r,s) \in \Gn} | \frac{\sqrt{n}|\hat A_{j;r}^{*(s,de)}|}{\widehat {s.e.}^* (j,r,s)}| \leq c ) - P( |\max_{i \in \mathcal{I}}  e_i^\top D_0^{-1} Z e_i| \geq c)|,
\end{align*}
where $Z$ is some appropriate Gaussian process and $D_0\in \R^{\G \times \G}$ specified later on. If we can show that the first term is $o(1)$ and the second $o_P(1)$ the assertion follows. We begin with the first term. 
Note that Lemma~\ref{lem.de.approx} gives us 
$
\max_{(j,r,s) \in {\Gn}} (\sqrt{n}|(\hat A_{j;r}^{(s,de)}-A_{j;r}^{(s)}))/\widehat {s.e.} (j,r,s)|=\max_{(j,r,s) \in {\Gn}} |(1/\sqrt{n}\sum_{t=d}^{n-1}   Y_{t;(r,j,s)})|+o_P(1),
$
where $Y_{t;(r,j,s)}= \linebreak e_{(s-1)p+r}^\top  (\Gammas)^{-1} W_t  \eps_{t+1;j}/s.e.(r,j,s)$.
Let $D_0=\operatorname{diag}(s.e(r_i,j_i,s_i), i \in \mathcal{I})$. We have $\var(D_0 1/\sqrt{n} \sum_{t=1}^n (Y_{t;(r_i,j_i,s_i)}, i \in \mathcal{I})^\top)=\Sigma_{T}=(\sigma_{T;i_1,i_2}, i_1,i_2 \in \mathcal{I}),$ where
$\sigma_{T;i_1,i_2}=\Sigma_{\eps,j_{i_1} j_{i_2}} e_{(s_{i_1}-1)p+r_{i_1}}^\top (\Gammas)^{-1}  e_{(s_{i_2}-1)p+r_{i_2}}$. To see this, note that $\{Y_t\}$ is an uncorrelated sequence and we have $1/n\cov(s.e(r_1,j_1,s_1) \sum_{t=1}^n  Y_{t;(r_1,j_1,s_1)},
s.e(r_2,j_2,s_2) \sum_{t=1}^n \linebreak Y_{t;(r_2,j_2,s_2)})=  1/n \sum_{t=1}^n \cov(e_{(s_1-1)p+r_1} (\Gammas)^{-1} W_t  \eps_{t+1;j_1},e_{(s_2-1)p+r_2} (\Gammas)^{-1} W_t \linebreak \eps_{t+1;j_2})=\sigma_{T;i_1,i_2}$.

Thus, for the process $Y:=\{(Y_{t,(r,j,s)})_{(r,j,s) \in \Gn}, t \in \Z\}$ we obtain by Lemma~\ref{lem.physical} that we have for some $\alpha>1/2$ $\Psi_{q,\alpha}=O(1), \Upsilon_{q,\alpha}=O(\min(\G^{1/q},\kp \kpp \log(\G)^{3/2}))$, and $\| \|Y_\cdot\|_\infty \|_{q,\alpha}=O(1)$. For $q$ large enough we have $O(\min(\G^{1/q},\kp \kpp \log(\G)^{3/2}))=O(\G^{1/q})$. Thus, remark 2 in \cite{zhang2017gaussian} can be applied which gives by Theorem 3.2 in \cite{zhang2017gaussian} an Gaussian approximation if $\G n^{1-q/2}(\log(\G n))^{3q/2}=o(1)$. This holds by Assumption~\ref{ass1}\eqref{ass1.6} and the limit on $b$. Applying now Theorem 3.2 in \cite{zhang2017gaussian} and using that the other terms are asymptotically negligible we obtain
$
\sup_{c \in \R} | P(\max_{(j,r,s) \in {\Gn}} \sqrt{n}|\hat A_{j;r}^{(s,de)}|/\widehat {s.e.} (j,r,s)) - P( |\max_{i \in \mathcal{I}}  e_i^\top D_0^{-1} Z e_i| \geq c)|=o(1),
$
where $Z=(Z_i, i \in \mathcal{I})\sim \mathcal{N}(0,\Sigma_T)$. 
In the following we show that
$
\sup_{c \in \R} | P^\ast( \max_{(j,r,s) \in G} \linebreak | (\sqrt{n}|\hat A_{j;r}^{*(s,de)}|)/(\widehat {s.e.}^* (j,r,s))| \leq c ) - P( |\max_{i \in \mathcal{I}}  e_i^\top D_0^{-1} Z e_i| \geq c)|=o_P(1).
$
Following the arguments used in the proof of Theorem~\ref{thm.bootstrap.lasso} we obtain $\|\Xi^{*(re)}-\Xi\|_l=O_{P}(\kp \sqrt{\gp/n})$, $\|\Sigma_{\eps}^{*}-\Sigma_{\eps}\|_\infty=O_{P}(\cp \sqrt{\gp/n})$ and $\|(\Gammahast)^{-1}-(\Gammas)^{-1}\|_\infty= \linebreak O_{P}(\kppp \sqrt{\gp/n})$, where the bootstrap analogue of the estimators is denoted by $*$. Hence, Lemma~\ref{lem.de.approx} implies
$\max_{(j,r,s) \in {\Gn}} (\sqrt{n}|\hat A_{j;r}^{*(s,de)}|)/(\widehat {s.e.}^* (j,r,s))=\max_{(j,r,s) \in {\Gn}}  1/\sqrt{n}\sum_{t=d}^{n-1} Y_{t;(r,j,s)}^*+o_P(1)$, where 
$Y_{t;(r,j,s)}^*=e_{(s-1)p+r}^\top (\Gammas)^{-1} W_t^* \linebreak  \eps_{t+1;j}^*/s.e.(r,j,s)$. Furthermore, we have
\begin{align*}
    \frac{1}{n}\cov^*&( \sum_{t=d}^{n-1} e_{(s_1-1)p+r_1}^\top (\Gammas)^{-1} W_t^* \eps_{t+1;j_1}^*, \sum_{t=d}^{n-1} e_{(s_2-1)p+r_2}^\top (\Gammas)^{-1} W_t^* \eps_{t+1;j_2}^*)\\
    &= (n-d)/n e_{(s_1-1)p+r_1}^\top (\Gammas)^{-1} \Gammah (\Gammas)^{-1} e_{(s_2-1)p+r_2}\hat \Sigma_{\eps;j_1j_2}\\
    &= e_{(s_1-1)p+r_1}^\top (\Gammas)^{-1} e_{(s_2-1)p+r_2}\Sigma_{\eps;j_1j_2}+o_P(1)=\sigma_{T;i_1,i_2}+o_P(1), 
\end{align*}
if $(j_1,r_1,s_1)=(j_{i_1},r_{i_1},s_{i_1})$ and $(j_2,r_2,s_2)=(j_{i_2},r_{i_2},s_{i_2})$. Thus, $\{Y_{t;(r,j,s)}^*, \linebreak (r,j,s) \in {\Gn}, t \in \Z \}$ possesses asymptotically the same autocovariance structure as $\{Y_{t;(r,j,s)}, (r,j,s) \in {\Gn}, t \in \Z \}$. This implies that Theorem 3.2 in \cite{zhang2017gaussian} can be applied in the same way as above to $\max_{(j,r,s) \in {\Gn}}  1/\sqrt{n}\sum_{t=d}^{n-1} Y_{t;(r,j,s)}^*$.  we obtain 
$
\sup_{c \in \R} | P^\ast( \max_{(j,r,s) \in G} | (\sqrt{n}|\hat A_{j;r}^{*(s,de)}|)/\widehat {s.e.}^* (j,r,s)| \leq c ) - P( |\max_{i \in \mathcal{I}}  e_i^\top D_0^{-1} Z e_i| \geq c)|=o_P(1).
$
\end{proof}

\begin{lem} \label{lem.finite.moment}
Let $X_t=\sum_{j=0}^\infty B_j \eps_{t-j},t \in \Z$, where the $ \eps_t$'s are i.i.d., $E\eps_1 \eps_1^\top =\Sigma$  $\sum_{j=0}^\infty \|B_j\|_2\leq C_1<\infty $ and $(E (v^\top \eps_0)^q)^{1/q}\leq C_2<\infty$ for all $\|v\|_2<\infty, C_1,C_2<\infty$. Then, $(E\|v^\top X_1\|^q)^{1/q}<\infty$ for $\|v\|_2<\infty$.
\end{lem}

\begin{lem}\label{lem.dagger}
Let $C$ be some positive definite matrix and $\beta_r^\dagger=e_r-I_{p;-r}\beta_{r;-r}$, where $\beta_{r;-r}=(I_{p;-r}^\top C I_{p;-r})^{-1}(I_{p;-r}^\top C e_r)$. Then, 
$
\beta_r^\dagger=(e_r^\top C^{-1} e_r)^{-1} C^{-1} e_r.
$
\end{lem}

{\bf Acknowledgments.}  The research of the first author was supported by the Research Center (SFB) 884 ``Political Economy of Reforms''(Project B6), funded by the German Research Foundation
(DFG). Furthermore, the authors acknowledge support by the state of Baden-W{\"u}rttemberg through bwHPC.

\bibliographystyle{apalike}
\bibliography{bib}

\clearpage
\begin{center}
{\large\bf Appendix A: Stacked VAR($d$) processes}
\end{center}
In this Appendix we discuss some properties of stacked VAR($d$) processes.
The VAR$(d)$ model, $X_t=\sum_{s=1}^d A^{(s)} X_{t-j} + \eps_t$,  can be written as a VAR$(1)$ model $\mathds{X}_t= \mathds{A} \mathds{X}_{t-1} + \mathds{U}_t$,  where  
$\mathds{X}_t = (X_t^\top,X_{t-1}^\top,\dots,X_{t-d}^\top)^\top,$ 

$$\mathds{A}= \begin{pmatrix}
A^{(1)}& A^{(2)} & \dots & A^{(d)}\\
I_p & 0 & \dots & 0\\
0 & \ddots & \ddots & \vdots\\
0 & \dots & I_p & 0 
\end{pmatrix}\in \R^{dp\times dp},
\mathds{U}_t=
\begin{pmatrix}
\eps_t \\
0\\
\vdots\\
0
\end{pmatrix}
\text{ and } 
\mathds{E}=
\begin{pmatrix}
I_p \\
0\\
\vdots\\
0
\end{pmatrix}\in \R^{dp\times p}.
$$
 $\{\mathds{X}_t\}$ is stable if $\{X_t\}$ is stable; see Section 2.1 in \cite{luetkepohl2007new}. Hence, in this case we also have the representation $\mathds{X}_t=\sum_{j=0}^\infty \mathds{A}^j \mathds{U}_{t-j}$ which gives  $\Gammas=\var(\mathds{X}_1)=\var((X_1,\dots,X_d)^\top)=\sum_{j=0}^\infty \mathds{A}^j \Sigma_{\mathds{U}}(\mathds{A}^\top)^j$, where $ \Sigma_{\mathds{U}}=\var(\mathds{U_1)}$. Note that a VAR(1) process with coefficient matrix $A$ is stable if $\rho(A)<1$, where $\rho(A)$ denotes the maximal absolute eigenvalue of $A$, which is equivalent to $\det(I_p-A(z))\not = 0$ for all $|z|\leq 1$. 
A VAR$(d)$ process is stable if $\det(I-\sum_{s=1}^d A^{(s)} z^s)\neq 0$ for all $|z|\leq 1$, see Section 2.1 in \cite{luetkepohl2007new}. Furthermore, this still holds for the stacked VAR$(1)$ process, that is, this process  is stable if $\det(I_p-\mathds{A}(z))\not = 0$ for all $|z|\leq 1$. However, for such a stacked coefficient matrix we may have that $\| \mathds{A} \|_2 \not \leq 1$;  see Lemma E.2 in \cite{basu2015}.
Furthermore, a stable and causal VAR$(d)$ process, possesses the representation $X_t=\sum_{j=0}^\infty B_j \eps_{t-j}$, $B_0=I$. For a VAR$(1)$ we have $B_j=A^j$.
Note that for a VAR$(d)$ process we have  $X_t=\sum_{j=0}^\infty B_j \eps_{t-j}= \sum_{j=0}^\infty  \mathds{E}^\top \mathds{A}^j \mathds{E} \eps_{t-j}$. Thus, $\|B_j\|=\|\mathds{E}^\top \mathds{A}^j \mathds{E}\| \leq \|\mathds{E}^\top\|\|\mathds{E}\| \|\mathds{A}^j\|=\|\mathds{A}^j\|$.


The following example refers to a VAR(1) process with   sparse matrices $A$ and  $\Sigma_\eps$  which possess a non-sparse lag zero autocovariance matrix $ \Gamma(0)$.

\begin{example}
Consider the VAR(1) process $X_t=A X_{t-1}+\eps_t$ with   $\eps_1 \sim(0,I_p)$ and where   $A_{1,1}=\delta$, $A_{j+1,j}=\delta,j=1,\dots,p$, $\delta \neq 0$, and  $A_{i,j}=0$  elsewhere. Then, $A$ is sparse and each row contains only one non-zero entry, namely  $\delta$. For  $\delta \in (-1,1)$, the VAR(1) process is causal. However, all elements of the   lag-zero autocovariance matrix $ \Gamma(0)$ are different from zero which also holds for the inverse. Notice that these  non-zero entries can be small  for large $p$ and for  the parameter $\delta$ not close  to the boundaries  $1$ respectively  $-1$. The magnitude of the non-zero entries can be increased for $\Gamma(0)$ or $\Gamma(0)^{-1}$, if additionally a non-diagonal $\Sigmaeps^{-1}$ or $\Sigmaeps$ is chosen, respectively. 
\end{example}


\clearpage
\begin{center}
{\large\bf Appendix B: Bias-corrected bootstrap percentage points}
\end{center}


Suppose we have data $ X_1, X_2, \ldots, X_n$ and we want to test $ H_0: \theta =\theta_0$ against $H_1:  \theta >\theta_0$, where $\theta$ is some parameter of interest. Suppose that to perform the test, a test statistic $T_n$ is used which rejects  $H_0$  if $ T_n > m_{n,\alpha}$, where $ m_{n,\alpha}$ is the upper $\alpha$-percentage point of the distribution of $ T_n$ under $H_0$.   A bootstrap-based test works by generating pseudo data $ X^*_1, X^*_2, \ldots, X^*_n$ under $ H_0$ and calculating the upper $\alpha$-percentage point $m_{n,\alpha}^*$ of the distribution of $ T_n^*$, where $T_n^*$ is the same statistic as $ T_n$ but based on $ X_1^*, X_2^*, \ldots, X_n^*$. That is, 
\begin{equation} \label{eq.bootperc}
     m_{n,\alpha}^*= F^{-1}_{T^*_n}\Big(\Phi(z_{\alpha})\Big),
     \end{equation}
where $\Phi$ denotes the distribution of function of the standard Gaussian random variable, $ \Phi(z_\alpha)=1-\alpha$, $F_{T^*_n} $ the distribution function of $ T_n^*$ and $F^{-1}_{T*_n}$ the quantile function of the same random variable.

Suppose now that  $ X_1^+, X_2^+, \ldots, X_n^+$ is a sample of pseudo 
random variables generated under $H_0$ and  using  a second bootstrap experiment which is based on $ X_1^*, X_2^*, \ldots, X_n^*$.  Let $ T_n^+$ be the test statistic of interest calculated using $ X_1^+, X_2^+, \ldots, X_n^+$. Assume now that the following assertions are true. 
\begin{enumerate}
    \item \ There exists a monotone increasing function $g$ such that 
    \begin{itemize}
\item[(i)] 
    $ g(T_n)-g(\theta_0) \sim N(z_0\sigma_0, \sigma^2_0)$.
    \item[(ii)]  \ $ g(T_n^*) -g(\theta_0) \sim N(-z_1\sigma_0, \sigma_0^2)$,
    (Recall that $T_n^*$ stems from a bootstrap experiment based on the observed data $X_1, X_2, \ldots, X_n$).
   \item[(iii)] \  $g(T_n^+) -g(\theta_0) \sim N(-z_2 \sigma_0, \sigma_0^2)$,  (Recall that $T_n^+$ stems from a bootstrap experiment based on the bootstrap sample  $X_1^*, X_2^*, \ldots, X_n^*$).
   \end{itemize}
    \item  It holds true  that 
    $ -z_0 + z_1 = -z_1 + z_2$. 
\end{enumerate}
The essential motivation behind the above assumptions is that the second bootstrap experiment which leads to the bootstrap sample $X_1^+, X_2^+, \ldots, X_n^+$, should  imitate the same bias  as the first bootstrap experiment leading to $X_1^*, X_2^*, \ldots, X_n^*$. 
If the above relations hold true, then it can be shown that a bias  corrected bootstrap percentage point which can be used for performing the bootstrap based test is given by
\[ m^{*(B)}_{n,\alpha}=F^{-1}_{T^*_n} \Big( \Phi\big(\sqrt{2}z_0 +  z_{\alpha}\big)\Big)=F^{-1}_{T^*_n} \Big( \Phi\big(\sqrt{2}\cdot \Phi^{-1}\big(P(T^+_n\leq T_n^*)\big) +  z_{\alpha}\big)\Big). \]
Notice that if no bias exist, that is if $z_0=0$ or $P(T^+_n\leq T^*_n) = 1/2$, respectively, then the standard bootstrap percentage point (\ref{eq.bootperc}) is obtained.  Furthermore,  we can estimate all quantities appearing in the expression for $ m^{*(B)}_{n,\alpha}$ without knowledge of the transformation  $g$. An  estimator of $z_0$  can be obtained  using  the following algorithm.
\begin{enumerate}
    \item[] {\it Step 1:} Generate a pseudo time series $X_1^*,\dots,X_n^*$ under the null using the bootstrap procedure  described  in Section~\ref{se.bootalgo}. 
    \item[] {\it Step 2:} Based on $X_1^*,\dots,X_n^*$, compute the test statistic $T_n^*$ and compute (under the null) the estimators $\hat A^{*(s)},s=1,\dots,d$, and $\hat \Sigma_\eps^{*}$.
    \item[] {\it Step 3:} Given the estimators $\hat A^{*(s)},s=1,\dots,d$, and $\hat \Sigma_\eps^{*}$, generate pseudo time  series $X_1^{+},\dots,X_n^{+}$ and compute the test statistic $T_n^{+}$.
    \item[] {\it Step 4:} Repeat {\it Step 3}  a number of times, say $B_2$ times, and estimate $z_0$ by $\hat z^+_0=\Phi^{-1}(\sum_{k=1}^{B_2} \ind{}\{T_{n,k}^{+}<T_n^*\}/B_2)$, where $ T^+_{n,k}$ denote  the value of the test statistic calculated using  the $k$th  pseudo time series, $k=1,2, \ldots, B_2$.
\end{enumerate}
To reduce the dependence of the above procedure on the initial time series $ X_1^*, X_2^*, \ldots, X_n^*$ appearing in Step 1, we apply the above algorithm to $200$  generated  bootstrap time series $ X_1^*, X_2^*, \ldots, X_n^*$ and average the estimates $ \hat{z}_0^+$ obtained.  
As long as the number $B_2$ of trials in Step 4 is large enough, the outcome is not affected considerably. 
In the numerical examples of this section we  use,  for computational reasons, $200$ repetitions and set  $B_2=60$. To reduce the computational burden one can apply the idea of \cite{giacomini2013warp}, that is a  good estimate of $ z_0$ can  also  be obtain by using  $B_2=1$ and  a larger number of repetitions.

\newpage 

\FloatBarrier
\clearpage
\begin{center}
{\large\bf Appendix C: Data description}
\begin{table}[ht]
\begin{tabular}{rllll}
  \hline
 & 1 & 2 & 3 & 4 \\ 
  \hline
1 & RPI & PAYEMS & M1SL & EXUSUKx \\ 
  2 & W875RX1 & USGOOD & M2SL & EXCAUSx \\ 
  3 & DPCERA3M086SBEA & CES1021000001 & M2REAL & WPSFD49207 \\ 
  4 & CMRMTSPLx & USCONS & TOTRESNS & WPSFD49502 \\ 
  5 & RETAILx & MANEMP & BUSLOANS & WPSID61 \\ 
  6 & INDPRO & DMANEMP & NONREVSL & WPSID62 \\ 
  7 & IPFPNSS & NDMANEMP & CONSPI & OILPRICEx \\ 
  8 & IPFINAL & SRVPRD & S.P.500 & PPICMM \\ 
  9 & IPCONGD & USTPU & S.P..indust & CPIAUCSL \\ 
  10 & IPDCONGD & USWTRADE & S.P.div.yield & CPIAPPSL \\ 
  11 & IPNCONGD & USTRADE & S.P.PE.ratio & CPITRNSL \\ 
  12 & IPBUSEQ & USFIRE & FEDFUNDS & CPIMEDSL \\ 
  13 & IPMAT & USGOVT & CP3Mx & CUSR0000SAC \\ 
  14 & IPDMAT & CES0600000007 & TB3MS & CUSR0000SAD \\ 
  15 & IPNMAT & AWOTMAN & TB6MS & CUSR0000SAS \\ 
   \hline
\end{tabular}
\caption{The first $60$ of the $p=124$  time series of FRED  used in the real-data example; a description as well as a grouping of the times series can be found in \cite{mccracken2016fred} and the references therein.}
\end{table}

\begin{table}[ht]
\begin{tabular}{rllll}
  \hline
 & 1 & 2 & 3 & 4 \\ 
  \hline
  16 & IPMANSICS & AWHMAN & GS1 & CPIULFSL \\ 
  17 & IPB51222S & HOUST & GS5 & CUSR0000SA0L2 \\ 
  18 & IPFUELS & HOUSTNE & GS10 & CUSR0000SA0L5 \\ 
  19 & CUMFNS & HOUSTMW & AAA & PCEPI \\ 
  20 & HWI & HOUSTS & BAA & DDURRG3M086SBEA \\ 
  21 & HWIURATIO & HOUSTW & COMPAPFFx & DNDGRG3M086SBEA \\ 
  22 & CLF16OV & PERMIT & TB3SMFFM & DSERRG3M086SBEA \\ 
  23 & CE16OV & PERMITNE & TB6SMFFM & CES0600000008 \\ 
  24 & UNRATE & PERMITMW & T1YFFM & CES2000000008 \\ 
  25 & UEMPMEAN & PERMITS & T5YFFM & CES3000000008 \\ 
  26 & UEMPLT5 & PERMITW & T10YFFM & UMCSENTx \\ 
  27 & UEMP5TO14 & AMDMNOx & AAAFFM & MZMSL \\ 
  28 & UEMP15OV & ANDENOx & BAAFFM & DTCOLNVHFNM \\ 
  29 & UEMP15T26 & AMDMUOx & TWEXMMTH & DTCTHFNM \\ 
  30 & UEMP27OV & BUSINVx & EXSZUSx & INVEST \\ 
  31 & CLAIMSx & ISRATIOx & EXJPUSx & VXOCLSx \\ 
   \hline
\end{tabular}
\caption{The last $64$ of the $p=124$  time series of FRED  used in the real-data example; a description as well as a grouping of the times series can be found in \cite{mccracken2016fred} and the references therein.}
\end{table}
\end{center}
\FloatBarrier

\clearpage 
\begin{center}
{\large\bf Appendix D: Simulation set-up}
\end{center}
Let $A_{BLOCK}^\xi=
\begin{pmatrix}
D & 0 \\
B & C
\end{pmatrix}$
and $
\Sigma_{BLOCK}=\begin{pmatrix}\Sigma_{11} & 0 \\ 0 & \Sigma_{22} \end{pmatrix},$
where $0$ denotes  a null matrix of appropriate dimensions and the submatrices $\Sigma_{11}$, $\Sigma_{22}$ $D$, $B$ and  $C$  are defined as follows. 
$$D=\operatorname{diag}\big< \xi,-.7,\xi,-.6,.6,0,0,0,0,0.2,0.5,-0.8,0,0\big>,$$ 
$$
B= 
\begin{pmatrix}
 0.8 & 0.2 & -0.4 & 0.0 & 0.0 & 0.0 & 0.0 & 0.0 & 0.0 & 0.0 & 0.0 & 0.0 & 0.0 & 0.0 \\ 
 0.0 & 0.6 & -0.7 & 0.0 & 0.0 & 0.8 & 0.0 & 0.0 & 0.0 & 0.0 & 0.0 & 0.0 & 0.0 & 0.0 \\ 
 0.0 & 0.0 & 0.0 & -0.9 & 0.0 & 0.0 & 0.0 & 0.0 & 0.0 & -0.6 & 0.0 & 0.0 & 0.0 & 0.0 \\ 
 0.0 & 0.0 & 0.0 & 0.8 & 0.0 & 0.0 & 0.2 & 0.0 & 0.0 & 0.0 & 0.0 & 0.0 & 0.0 & 0.0 \\ 
 0.0 & 0.7 & 0.0 & 0.0 & 0.0 & -0.3 & 0.0 & 0.0 & 0.0 & 0.0 & 0.0 & 0.0 & 0.0 & -0.7 \\ 
 0.3 & 0.0 & 0.0 & 0.0 & 0.0 & 0.0 & 0.0 & 0.0 & 0.9 & 0.0 & 0.0 & 0.0 & 0.0 & 0.0,
\end{pmatrix}
$$
$$
C=\begin{pmatrix}
 \xi & 0.0 & 0.0 & 0.0 & 0.0 & 0.0 \\ 
0.0 & 0.0 & 0.0 & 0.3 & 0.0 & 0.0 \\ 
0.0 & 0.0 & 0.0 & 0.0 & -0.3 & 0.0 \\ 
0.6 & 0.0 & 0.0 & 0.0 & 0.0 & 0.0 \\ 
0.0 & 0.0 & 0.6 & 0.0 & 0.0 & 0.0 \\ 
0.0 & 0.0 & 0.0 & 0.0 & 0.0 & \xi \\ 
  \end{pmatrix}
  \quad \mbox{and}  \quad 
\Sigma_{22}=\begin{pmatrix}
1.00 & 0.25 & 0.25& 0.25 & 0.25& 0.25 \\ 
0.25 & 1.00 & 0.00 & 0.00 & 0.00 & 0.00 \\ 
0.25 & 0.00 & 1.00 & 0.00 & 0.00 & 0.00 \\ 
0.25 & 0.00 & 0.00 & 1.00 & 0.00 & 0.00 \\ 
 0.25 & 0.00 & 0.00 & 0.00 & 1.00 & 0.00 \\ 
 0.25 & 0.00 & 0.00 & 0.00 & 0.00 & 1.00 \\ 
  \end{pmatrix}
  $$
  $$
\Sigma_{11}=
\begin{pmatrix}
 1.0 & 0.5 & 0.0 & 0.0 & 0.0 & 0.0 & 0.0 & 0.0 & 0.0 & 0.0 & 0.0 & 0.0 & 0.0 & 0.0 \\ 
 0.5 & 1.0 & 0.5 & 0.0 & 0.0 & 0.0 & 0.0 & 0.0 & 0.0 & 0.0 & 0.0 & 0.0 & 0.0 & 0.0 \\ 
 0.0 & 0.5 & 1.0 & 0.5 & 0.0 & 0.0 & 0.0 & 0.0 & 0.0 & 0.0 & 0.0 & 0.0 & 0.0 & 0.0 \\ 
 0.0 & 0.0 & 0.5 & 1.0 & 0.5 & 0.0 & 0.0 & 0.0 & 0.0 & 0.0 & 0.0 & 0.0 & 0.0 & 0.0 \\ 
 0.0 & 0.0 & 0.0 & 0.5 & 1.0 & 0.0 & 0.0 & 0.0 & 0.0 & 0.0 & 0.0 & 0.0 & 0.0 & 0.0 \\ 
 0.0 & 0.0 & 0.0 & 0.0 & 0.0 & 1.0 & 0.0 & 0.0 & 0.0 & 0.0 & 0.0 & 0.0 & 0.0 & 0.0 \\ 
 0.0 & 0.0 & 0.0 & 0.0 & 0.0 & 0.0 & 1.0 & 0.0 & 0.0 & 0.0 & 0.0 & 0.0 & 0.0 & 0.0 \\ 
 0.0 & 0.0 & 0.0 & 0.0 & 0.0 & 0.0 & 0.0 & 1.0 & 0.0 & 0.0 & 0.0 & 0.0 & 0.0 & 0.0 \\ 
 0.0 & 0.0 & 0.0 & 0.0 & 0.0 & 0.0 & 0.0 & 0.0 & 1.0 & 0.0 & 0.0 & 0.0 & 0.0 & 0.0 \\ 
 0.0 & 0.0 & 0.0 & 0.0 & 0.0 & 0.0 & 0.0 & 0.0 & 0.0 & 1.0 & -0.5 & 0.0 & 0.0 & 0.0 \\ 
 0.0 & 0.0 & 0.0 & 0.0 & 0.0 & 0.0 & 0.0 & 0.0 & 0.0 & -0.5 & 1.0 & -0.5 & 0.0 & 0.0 \\ 
 0.0 & 0.0 & 0.0 & 0.0 & 0.0 & 0.0 & 0.0 & 0.0 & 0.0 & 0.0 & -0.5 & 1.0 & 0.0 & 0.0 \\ 
 0.0 & 0.0 & 0.0 & 0.0 & 0.0 & 0.0 & 0.0 & 0.0 & 0.0 & 0.0 & 0.0 & 0.0 & 1.0 & 0.0 \\ 
 0.0 & 0.0 & 0.0 & 0.0 & 0.0 & 0.0 & 0.0 & 0.0 & 0.0 & 0.0 & 0.0 & 0.0 & 0.0 & 1.0 \\
  \end{pmatrix}
  $$

\clearpage 
\begin{center}
{\large\bf Appendix E: Additional proofs}
\end{center}
\begin{proof}[Proof of Lemma~\ref{lem.physical}]
Note that $W_t=\A W_{t-j} +\mathds{E} \eps_{t}=\sum_{j=0}^\infty \A \mathds{E} \eps_{t-j}$ and by Assumption ~\eqref{ass1.2} $\|A^k\|_2=O(\lambda^k),\lambda<1$. Such a VAR(1) process possesses the functional dependence measure $\delta_{k,q,i}\leq C_q \|{\A}^k\|_2=O(\lambda^k)$ for every $i$ as given by Example 2.2 in \cite{chen2013covariance}. Thus, for $\{W_t\}$ and some $\alpha>1/2-1/q$ we obtain $\Psi_{q,\alpha}<\infty$ and $\Upsilon_{q,\alpha}=O(p^{1/q})$. We obtain $\| \|Y_\cdot\|_\infty \|_{q,\alpha}=O(\kp)$ due to  $\|\A\|_\infty=O(\kp)$. 
Since $\|(\Gammas)^{-1}\|_2<\infty$ and $(e_r^\top ((\Gammas)^{-1} W_t W_t^\top -I_{dp})e_r)_{r=1,\dots,dp}$ is a nonlinear transformation which  preserves the functional dependence, see \cite{wu2005nonlinear}, these results can be transferred and we obtain 1.
Similarly, we obtain 2. Note here that $\|(\Gammas)^{-1} \A\|_\infty=O(\kpp \kp)$.
\end{proof}

\begin{proof}[Proof of Lemma~\ref{lem.finite.moment}]
We have $v^\top X_1=\sum_{j=0}^\infty v B_j \eps_{1-j}=\sum_{j=0}^\infty \xi_j=\sum_{j=0}^n \xi_j+Z_n$ and $\{\xi\}$ is an i.i.d. sequence. Thus, $E |v^\top X_1|^q\leq C_q E(\sum_{j=0}^n \xi_j)^q+ E Z_n^q$. Since $B_j$ is summable and $(E (v^\top \eps_0)^q)^{1/q}<\infty$, $E Z_n^q$ can be made arbitrary small. Furthermore, by Rosenthal inequality as in Example 2.2 in \cite{zhang2018gaussian} we obtain $E(\sum_{j=0}^n \xi_j)^q\leq C_q \linebreak  (\sum_{j=0}^n E(\xi_j)^2)=C_q (\sum_{j=0}^n v^\top B_j \Sigma B_j^\top v)^q<\infty$.
\end{proof}

\begin{proof}[Proof of Lemma~\ref{lem.concen}]
Let $Z_{t;r}=e_r^\top [(\Gammas)^{-1} W_t W_t^\top -I_{dp}]v$. We have $E Z_{t;r}=0$.
Furthermore, $\|Z_{t;r}+1\|_{E,q/2}^{q/2}=E(e_r\top (\Gammas)^{-1} \sum_{j=0}^\infty \A^j \U_{t-j} \sum_{s=0}^\infty \U_{t-s}^\top (\A^s)^\top v)^q/2\leq \linebreak E(e_r\top (\Gammas)^{-1} \sum_{j=0}^\infty \A^j \U_{t-j} )^q+E(\sum_{s=0}^\infty v \A^s \U_{t-s})^q <\infty$ by Lemma~\ref{lem.finite.moment} and  due to $\|\A^j\|_2\leq \|\Gammas\|_2<C$. Furthermore, Lemma~\eqref{lem.physical} implies $\nu_{q/2}<\infty$. This allows us to use Nagaev's inequality for physical dependent processes, see Theorem 2 in \cite{liu2013probability} and their remark below Theorem 2. Hence, we obtain for some $M>0$
\begin{align*}
    P&(\max_{r=1,\dots,dp} | \frac{1}{n} \sum_{t=d}^{n-1} Z_{t;r} | > \sqrt{\log (p) /n}  M) \leq p \max_{r} P(| \frac{1}{n} \sum_{t=d}^{n-1} Z_{t;r} | > \sqrt{\log (p) /n}  M) \\
    &\leq p \max_r \Big[c_q \frac{n(\nu_{q/2}^{q/2+1}+\|Z_{t;r}\|_{E,q/2}^{q/2}}{ (n \log(p))^{q/4} (M )^{q/2}}+c_q^\prime \exp(-c_q \frac{\log(p) M^2 }{\nu_{q/2}^{2+4/q}}+2 \exp(-c_q \frac{\log(p) M^2 }{\|Z_{t;r}\|_{E,2}^2})\Big]\\
    & \leq C \frac{np C}{n^{q/4} \log^{q/4}(p) M^{q/2}} + C \exp(-\log(p) [C M^2 -1]).
\end{align*}
Since $(np)(n^{q/4} \log^{q/4}(p))=O(1)$, there exists for every $\eps>0$ an $M>0$ such that the probability is smaller $\eps$.
\end{proof}

\begin{proof}[Proof of Lemma~\ref{lem.dagger}]
Note that $I_p=e_r e_r^\top +I_{p;-r}I_{p;-r}^\top$. Thus, we have
\begin{align*}
(e_r-I_{p;-r} & \beta_{r;-r})(e_r^\top C^{-1} e_r)
=e_r e_r^\top C^{-1} e_r -I_{p;-r} (I_{p;-r}^\top C I_{p;-r})^{-1}(I_{p;-r}^\top C e_r) e_r^\top C^{-1} e_r\\
&=C^{-1} e_r-I_{p;-r}^\top (I_{p;-r}^\top C I_{p;-r})^{-1}[I_{p;-r}^\top C I_{p;-r}I_{p;-r}^\top C^{-1} e_r+I_{p;-r}^\top C e_r e_r^\top C^{-1} e_r]\\
&=C^{-1} e_r-I_{p;-r}^\top (I_{p;-r}^\top C I_{p;-r})^{-1} \underline{0}.
\end{align*}
\end{proof}

\end{document}